\documentclass[10pt,twocolumn,twoside]{IEEEtran}
\usepackage{color,graphics,graphicx,amsmath,amsfonts,
amssymb,times,setspace,cite,fancyhdr,epstopdf}
\usepackage{subfigure}
\usepackage[active]{srcltx}

\usepackage[english]{babel} % English language/hyphenation
\usepackage{nomencl}
\usepackage{adjustbox}
\usepackage[dvipsnames]{xcolor}
\usepackage{enumerate}% http://ctan.org/pkg/enumerate

\usepackage{array}
\usepackage{makecell}
\usepackage{booktabs}
\usepackage{multirow}
\usepackage{makecell}
\usepackage[utf8]{inputenc}
\usepackage{algpseudocode,algorithm,algorithmicx}
\makenomenclature

\def\BibTeX{{\rmfamily B\kern-.05em{\scshape i\kern-.025em b}\kern-.08em \TeX}}

\newtheorem{proposition}{Proposition}

\newcommand{\ben}{\begin{equation}}
\newcommand{\een}{\end{equation}}
\newcommand{\beq}{\begin{eqnarray}}
\newcommand{\eeq}{\end{eqnarray}}

\renewcommand\floatsep{0mm}
\renewcommand\textfloatsep{2mm}

\begin{document}

\title{Fully Decentralized Peer-to-Peer Community Grid with Dynamic and Congestion Pricing}
\author{Hien Thanh Doan, Truong Hoang Bao Huy, Daehee Kim~\IEEEmembership{Member,~IEEE} and Hongseok Kim~\IEEEmembership{Senior Member,~IEEE}}
%\author{Hongseok Kim,~\IEEEmembership{Senior Member,~IEEE}, Joohee Lee, Shahab Bahrami  and Vincent W.S. Wong~\IEEEmembership{Fellow,~IEEE}}
%\thanks{Seugnhyoung Ryu, Hyungeun Choi, and Hongseok Kim are with the Department of Electronic Engineering, Sogang University, 35 Baekbeom-ro, Mapo-gu, Seoul, 04107 , Korea
%e-mail: hongseok@sogang.ac.kr.}

\maketitle
\begin{abstract}
Peer-to-peer (P2P) electricity markets enable prosumers to minimize their costs, which has been extensively studied in recent research. However, there are several challenges with P2P trading when physical network constraints are also included. Moreover, most studies use fixed prices for grid power prices without considering dynamic grid pricing, and equity for all participants. This policy may negatively affect the long-term development of the market if prosumers with low demand are not treated fairly. An initial step towards addressing these problems is the design of a new decentralized P2P electricity market with two dynamic grid pricing schemes that are determined by consumer demand. Futhermore, we consider a decentralized system with physical constraints for optimizing power flow in networks without compromising privacy. We propose a dynamic congestion price to effectively address congestion and then prove the convergence and global optimality of the proposed method. Our experiments show that P2P energy trade decreases generation cost of main grid by 56.9\% compared with previous works. Consumers reduce grid trading by 57.3\% while the social welfare of consumers is barely affected by the increase of grid price.
\end{abstract}
\begin{keywords}
Decentralized electricity market, peer-to-peer energy trading, peer-to-grid, physical constraint, main grid dynamic pricing
\end{keywords}
%\begin{IEEEkeywords}
%Smart Grid; load profile; clustering; deep learning; convolutional autoencoder; convolutional neural network; yearly load profile
%\end{IEEEkeywords}

% \fi
\IEEEpeerreviewmaketitle

\section{Introduction}
In recent years, there have been significant interests in research on distributed energy resources (DERs), such as renewable energy resources (RESs) in power systems. Each node in a distribution network is equipped with smart devices capable of exchanging information and switching to appropriate software devices, such as smart meters and energy management systems (EMSs). It enables flexible scheduling, monitoring, and sharing of energy usage information in a distribution network, and encourages prosumers to participate in energy trading on a proactive basis. A variety of energy market development projects have been implemented in distribution systems, such as SonnenCommunity in Germany \cite{SonnComm}, Brooklyn in the USA \cite{BroMicr}, and Piclo in the UK \cite{Piclo} to enable prosumers to utilize their DERs. All of these are expected to contribute to managing and optimizing energy resources in the future.

In practice, prosumers often buy/sell electricity through retailers due to the daily fluctuation of grid price. The safety and the stability of the power system is also a major concern when transmitting electricity directly. However, this can change in the next generation of the power grid for a number of reasons, such as the development of EMS that enables each prosumer to manage its own risks under grid price fluctuations. Besides, with the development of a management system, each node on the power system can manage its own safety and stability, as well as service fees \cite{soto2021peer}. It can also operate with other nodes in the whole system without depending on retailers. Therefore, this study aims at designing the future electricity market that allows prosumers to transact energy directly with other prosumers and also with the main grid.
 
A peer-to-peer electricity market (P2PEM) is a new type of market that allows surpluses and deficits among network peers to be directly traded. The P2PEM provides cost-saving, autonomy, transparency, and competition \cite{aggarwal2021survey} to each participant. In order to provide cost-saving opportunities, \cite{abeygunawardana2021grid} indicates that appropriate P2PEM should encourage prosumers to remain involved. Therefore, it is essential to devise a long-term market mechanism that supports fairness and incentives for energy trading among participants.

Early studies have considered different P2P trading negotiation mechanisms such as centralized, decentralized, and auction-based approaches \cite{azim2019feasibility}. The centralized mechanism has a centralized transaction process and an information-sharing manner. The auction-based mechanism \cite{doan2021peer} is an approach in which prosumers relay the information to an aggregator to maximize their profits. However, relaying data to the central entity may leak the concern of privacy. One solution could be using a decentralized algorithm that uses limited information exchange and matches prosumers directly \cite{li2021data}. In the decentralized mechanism, each prosumer negotiates with another prosumer to find an optimum solution based on their preferences, which is the focus of this paper.

There has been a variety of research on decentralized P2P trading. In \cite{khorasany2021lightweight}, a market is proposed for grid-connected prosumers. In \cite{van2020integrated,doan2022optimal}, an optimization model was presented for prosumers, using an application of the alternating direction method of multipliers (ADMM) algorithm for matching prosumers. In \cite{doan2022optimal}, the authors proposed an electricity market by considering optimization problem separately and independently each time. Time-coupling constraint in multiple time is studied in \cite{fernandez2021bi} for battery. Although all of the above articles focused on grid-connected prosumers, they all have the same limitation in that the price offered by main grid is the predetermined one, which will be called a \textit{fixed} price hereafter. Of course, the time-of-use (TOU) price changes during a day, but not directly related with amount of the load in community grid in that time slot. This necessitates new grid pricing mechanisms between the main grid and the community grid to encourage or discourage electricity consumption dynamically, which is called \textit{dynamic grid pricing} hereafter.

Indeed, the dynamic grid pricing was incorporated in several studies. In \cite{luo2021multiple}, an energy trading scheme for different stakeholders at multiple levels was proposed. In \cite{chen2020dynamic,li2021real}, demand response was applied to minimize energy costs based on the preferences of market participants and to reduce the effects of dynamic grid price. The authors suggested a new objective function for prosumers to improve prosumer's utility \cite{li2021real} or satisfaction \cite{jafari2020fair}. Also in \cite{jafari2020fair}, the authors considered an equitable allocation  of profits among microgrids (MGs). However, these studies only considered social welfare maximization (SWM) or energy cost minimization in a \textit{virtual} layer. As a consequence, physical constraints such as line losses, voltage variation and congestion are neglected, so is the fairness. 

Several works have focused on considering these physical constraints in P2PEM; a decentralized market was studied using Nash bargaining to solve alternating current optimal power flow (AC-OPF); a branch-flow model was relaxed in second-order cones to resolve the non-convex problem with guaranteed exactness for radial grid topologies \cite{low2014convex}. However, the Nash bargaining based studies \cite{kim2019direct,zhong2020cooperative} can be impractical because it requires an honest report about the increased revenue of each prosumer. A newly added study in \cite{tofighi2022decentralized} attempted to solve both privacy concerns and congestion management but neglected reactive power and voltage constraints.

%% add more references with nash 17,18, NNNNNNote:  compared with previous quadratic and logarithmic utility functions, the new utility function can not only reduce the user’s electricity consumption and the supplier’s cost can but also improve the user’s utility and the total social welfare, which also indicates that the new utility function is effective in establishing the real-time pricing model of smart grid.  

In this regard, we address the above concerns by answering the following fundamental questions: 1) how to design an electricity market that maximizes social welfare while the main grid adjusts the price based on the community load, 2) what are the criteria for designing dynamic grid price in P2P trading, considering its impact on sustainability and stability in market development, 3) how to minimize the cost incurred in trading while keeping the privacy of each node and considering physical constraints, 4) how to impose a P2P trading fee to manage congestion, which is called \textit{congestion pricing} hereafter. The main contributions of this paper are as follows:

\begin{itemize}
\item We propose a decentralized P2P community market enabling dynamic pricing in a grid-connected environment with physical constraints. We conduct a mathematical formulation of the proposed P2P market with the objective of SWM, which is solved in a decentralized manner using ADMM.

\item A novel two-stage operation model is proposed to guarantee physical constraints and achieve optimality. In the first stage, decentralized P2P trading in a virtual layer is performed as an iterative procedure using ADMM. Upon reaching a virtually optimal P2P matching, the second stage finds a physical way that enables the P2P transactions with minimum power  losses and voltage fluctuations by solving a decentralized AC-OPF using ADMM. We prove that the proposed congestion pricing maximizes social welfare by coordinating the virtual layer and the physical layer iteratively while maintaining the privacy of each node in the power system, see \textit{Proposition~\ref{prop:1}}. Fig.~\ref{fig:overview} summarizes the proposed framework encompassing the virtual layer and the physical layer.

\item We present and investigate two grid pricing schemes, unique price scheme (UPS) and differential price scheme (DPS), to guarantee fairness among participants which encourage prosumers to participate in the market actively. It is important to note that fairness is essential for the long-term market mechanism and active participation of prosumers. 

\item Our proposed decentralized P2P market mechanism is analyzed in detail with various case studies under realistic configurations. Compared to the existing approaches, the proposed method is shown to decrease the generation cost of main grid by 56.9{\%}. In addition, the proposed method reduces the energy consumption from non-RES, such as main grids, by 57.3\%; thus, it is more environmentally friendly than other methods. Finally, the AC-OPF solution ensures the optimal power flow of the power systems, while a small number of rounds handle congestion.
\end{itemize}

The remainder of this paper is organized as follows. In Section~\ref{sec:system}, we present a system model for our work. In Section~\ref{sec:dynamic_price}, we formulate an energy trading pricing model and a congestion penalty rule. Section~\ref{sec:2-stage} describes the decentralized market within the context of a power system and examines the operational model. Section~\ref{sec:performance} presents the results and analysis of the case studies while Section~\ref{sec:conclusion} concludes this study.
\begin{figure}[t]
\centering
\includegraphics[width=\columnwidth]{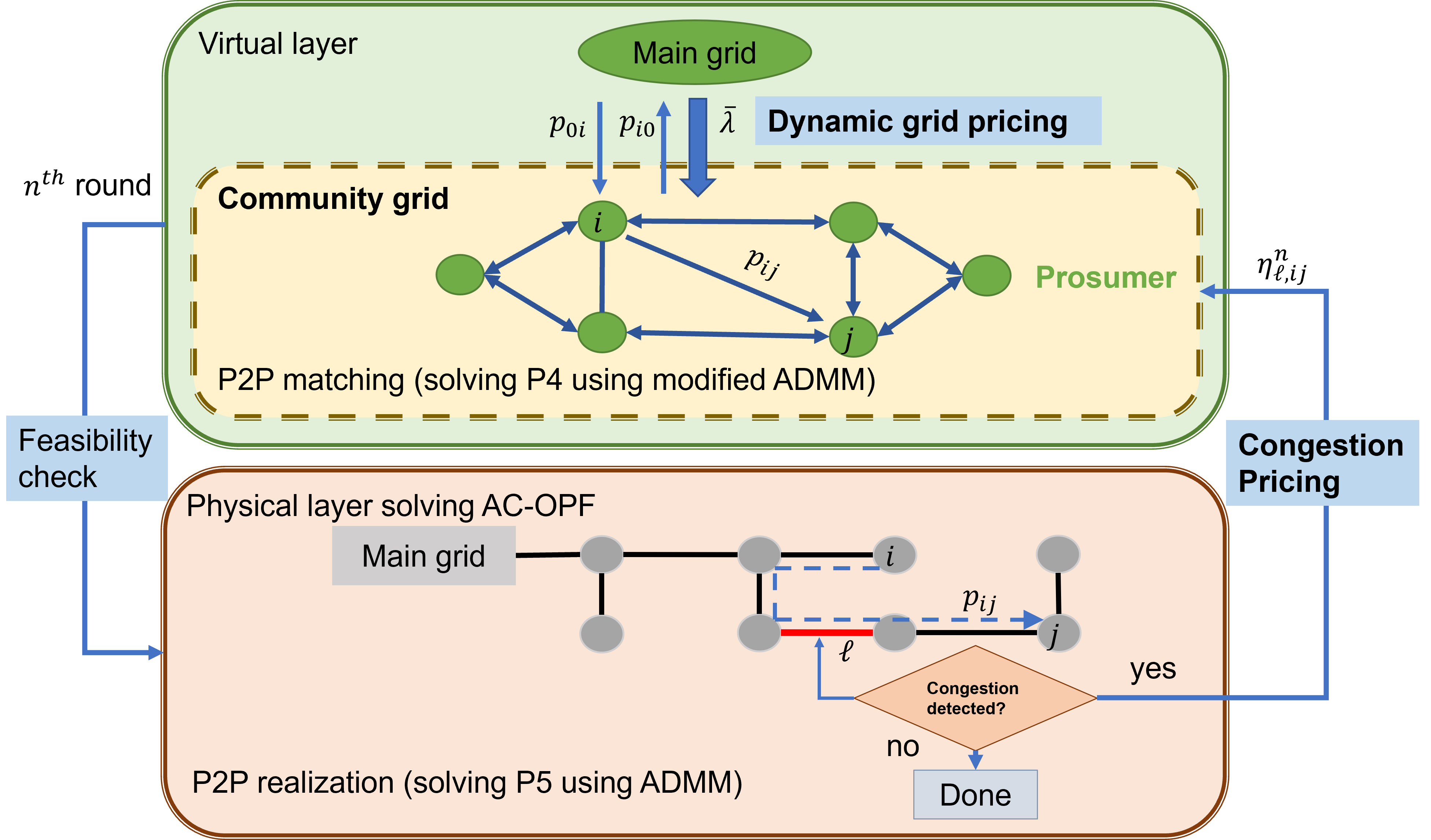}
\caption{Overview of electricity market systems.}
\label{fig:overview}
\end{figure}
\section{System Model}\label{sec:system}
\subsection{Structure of community grid}
Prosumers participating in a P2P market are assumed to be non-strategic and rational. A prosumer can be a consumer when generation capacity is less than demand. Consumers purchase energy from producers or the main grid. On the other hand, prosumer can be a producer when generation exceeds demand. Producers sell energy to consumers or the main grid. To handle all these activities, each prosumer is assumed to have a demand response (DR)-enabled smart meter that can record prosumer generation data, demand, and thus manage the amount of sold or purchased energy in P2P market. The market is an hourly ahead market where prosumers negotiate trades for the next time slot. Since we focus on a single time slot, we omit the time slot index.
%%%%%% new section 
\subsection{Distribution network model}
Let us consider a low-voltage (LV) distribution grid given by the undirected connected graph, $\mathcal{G}=(\mathcal{N},\mathcal{L})$ as shown in Fig.~\ref{fig:power_flow}. Here, $\mathcal{N}$ denotes a set of nodes indexed by $i=0,1,\ldots,|\mathcal{N}|$ and $\mathcal{L}$ denotes a set of lines connecting those nodes indexed as $\ell=1,\ldots,|\mathcal{L}|$. The slack bus (root node) has a zero index. Each node $i$ has a parent (ancestor), denoted as $\mathcal{A}_i$. We consider a radial distribution network, and thus each node has only one parent. A set of children of node $i$ is denoted by $\mathcal{C}_i$, indexed by $k=1,\ldots,|\mathcal{C}_i|$. Since we are considering a radial network, each line $\ell\in\mathcal{L}$ can be uniquely indexed by its connected child. Hence, we use the same notation $i$ to denote a node or a connected line towards the root node, unless $\ell$ needs to be specified for congestion pricing later.
\begin{figure}
\centering
\includegraphics[width=\columnwidth]{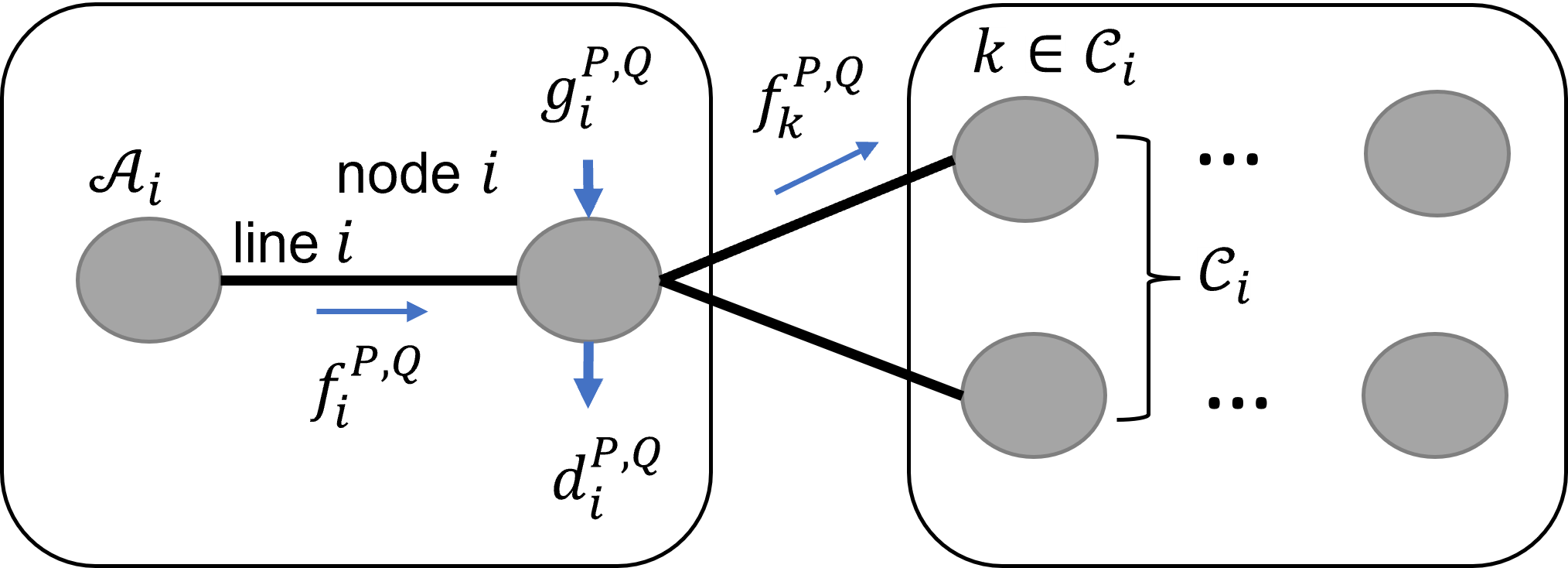}
\caption{A distribution network for P2P electricity market and power flow between participants.}
\label{fig:power_flow}
\end{figure}

Let $I_i$ be the current flow from the parent $\mathcal{A}_i$ to node $i$. In addition, $d_i^P$ and $g_i^P$ represent active power of the consumer and producer, while $d_i^Q$ and $g_i^Q$ represent the reactive powers of the consumer and producer, respectively. $f_i^P$ and $f_i^Q$ represent the active and reactive power flows, respectively, of line $i$. Let $V_i$ be the voltage at node $i$. Then, the squared voltage at node $i$ is represented by $v_i=|V_i|^2$, and the squared current is represented by $l_i=|I_i|^2$. $v_{i}^{\min}$ and $v_{i}^{\max}$ are the minimum and the maximum squared voltage angles. Resistance $r_i$ and reactance $x_i$ are characterized for each line. $S^{\max}_{i}$ is the maximum capacity of line $i$. In an AC-OPF of the radial network, these quantities can be related using a LinDistFlow \cite{mieth2019distribution}. The equations are as follows. For $\forall i\in \mathcal{N}$, we have
\begin{subequations}
\label{eq:1}
\begin{align}
f_i^P+g_i^P-\sum_{k\in \mathcal{C}_i}(f_k^P+r_{k}l_{k}) &=& d_i^P, \label{eq:1a}\\
f_i^Q+g_i^Q-\sum_{k\in \mathcal{C}_i}(f_k^Q+r_{k}l_{k}) &=& d_i^Q, \label{eq:1b}\\
v_i+2(r_{i}f_k^P+x_{i}f_k^Q)+l_{i}(r_{i}^2+x_{i}^2) &=& v_{\mathcal{A}_i}, \label{eq:1c}\\
(f_{i}^P)^2+(f_{i}^Q)^2 &\leq& (S^{\max}_{i})^2. \label{eq:capacity limit} \\
\frac{(f_{i}^P)^2+(f_{i}^Q)^2}{v_{i}} &=& l_{i}, \label{eq:1d} \\
v_{i}^{\min} \leq v_{i} \leq v_{i}^{\max}. \label{eq:1f}
\end{align}
Equation (\ref{eq:1}) defines the distribution network in the node variable set \{$d_i^P,g_i^P,d_i^Q,g_i^Q,v_i,l_i|i\in\mathcal{N}$\}. Equation (\ref{eq:1d}) being a non-convex constraint, can be relaxed to an inequality using a second-order cone \cite{low2014convex},
\begin{align}
||2f_{i}^P,2f_{i}^Q,v_{i}-l_{i}||_2 \leq v_{i}+l_{i}. \label{eq:non-linear-constraint}
\end{align}
\end{subequations}
%%%%%% new section 
\subsubsection{Cost function}
Producers engage in the market to maximize their benefits, where they try to sell their energy at a beneficial price to consumers or the main grid. Let $C_i(g_{i})$ represent the cost when a prosumer $i$ generates an amount of energy $g_{i}$ \cite{mehdinejad2022peer}.\footnote{We use $g_i$ instead of $g_i^P$ for notational simplicity hereafter.} The formula can be expressed as
\begin{subequations}
\label{eq:Producer cost function with constraint}
\begin{align}
C_i(g_{i})=b_{i}g_{i}+a_{i}g_{i}^2, \label{eq:Producers Cost Function} \\
g_{i}^{\min}\leq g_{i}\leq g_{i}^{\max}, & \quad \forall i\in \mathcal{N}, \label{eq:Producers energy limit}
\end{align}
\end{subequations}
where $a_i\geq0$ is to the dynamic cost of energy generation in \textdollar$/MWh^2$, $b_i>0$ is to the producer's minimum selling price in \textdollar$/MWh$, and $g_{i}^{\min}$ and $g_{i}^{\max}$ represent the minimum and the maximum amounts of energy generation in $MWh$.
%%%%%% new section 
\subsubsection{Utility function}
Consumers are willing to pay money to purchase energy from producers or the main grid. The responses of different consumers to various scenarios can be modeled using the concept of utility function. As done in \cite{yang2022hierarchical}, each consumer $i$ has level of satisfaction when it consumes $d_i$ amount of energy.\footnote{We use $d_i$ instead of $d_i^P$ for notational simplicity hereafter.}
\begin{subequations}
\label{eq:Consumer cost function with constraint}
\begin{align}
U_i(d_i) =
  \begin{cases}
    \beta_{i}d_{i}-\alpha_{i}{d_i}^2 & \quad \text{if } 0 \leq d_{i}\leq \frac{\beta_{i}}{2\alpha_{i}}\\
    \frac{\beta_{i}^2}{4\alpha_{i}}  & \quad \text{if } d_{i} > \frac{\beta_{i}}{2\alpha_{i}}
  \end{cases},
\label{eq:Consumer Utility Function}
\end{align}
\begin{equation}
d_{i}^{\min}\leq d_{i}\leq d_{i}^{\max}, \forall i\in \mathcal{N}, \label{eq:Consumer energy limit}    
\end{equation}
\end{subequations}
where $\alpha_i>0$ is consumer $i$'s satisfaction with energy consumption in \textdollar$/MWh^2$, and $\beta_i>0$ is the consumer's maximum buying price in \textdollar$/MWh$. Consumer $i$ can buys energy $d_{i}$ from producers and main grid, where $d_{i}^{\min}$ and $d_{i}^{\max}$ represent the minimum and the maximum required electricity in $MWh$ of the consumer $i$, respectively. More specifically, $d_{i}^{\min}$ represents a realistic load demand for fixed loads, whereas $d_{i}^{\max}$ indicates the load demand for fixed and flexible loads.
%%%%%%  new subsection 
\subsubsection{Main grid cost function}
Main grid can supply or absorb power at any given time due to the mismatch between generation and demand in P2P markets. The cost of energy trading with main grid can be modeled as a quadratic function \cite{luo2021multiple,chen2020dynamic,li2021real,jafari2020fair}. However, unlike the previous works, the cost function of the proposed method is calculated by coefficient parameters $a_0\geq0, b_0>0$ and predefined by main grid at each time slot. The cost of main grid to supply $p_{0}$ amount of power is given by 
\begin{equation}
% \label{eq:PG cost function with constraint} 
% \begin{align}
C_0(p_0)=b_{0}p_{0}+a_{0}p_{0}^2, \label{eq:main grid Cost Function} 
% \end{align}
\end{equation}
where $a_0$ denotes dynamic cost in \textdollar$/MWh^2$, and $b_0$ represents the minimum price in \textdollar$/MWh$. Since the node $0$ is a slack bus, we assume that $p_0$ does not have the maximum or the minimum constraints. 
%%%%%% new section 
\section{Problem Formulation}\label{sec:dynamic_price}
\subsection{Grid Dynamic Pricing} \label{subsec:PG_model}
According to (\ref{eq:Producers Cost Function}), (\ref{eq:Consumer Utility Function}), (\ref{eq:main grid Cost Function}) the objective of the electricity market in virtual layer can be expressed as

\vspace{2mm}
\noindent \textbf{P1: Utility Maximization}
\begin{align}
\max \sum_{i\in\mathcal{N}} [U_i(d_i)-C_i(g_i)]-C_0(p_0), \label{eq:objective of the electricity market}
\end{align}
and \textbf{P1} can be solved by decentralized optimization, as in previous works \cite{tofighi2022decentralized,kim2019p2p,khorasany2019decentralized}. In this study, however, we reformulate (\ref{eq:objective of the electricity market}) by converting the quadratic cost function of $C_0(p_0)$ of the main grid into a linear function, which has several benefits compared to other studies. First, in practice, Feed-in Tariff (the selling price to the grid) is fixed at all times of day \cite{eseye2017grid} or fixed at each time slot. In the context of this study, this corresponding to $b_0>0$ and $a_0=0$. Hence, the main grid can publicize the selling price $b_0$ before the P2P electricity market starts, and producers can solve these problems without communicating the main grid. Accordingly, the previous studies lack efficiency, primarily because producers are required to communicate with the grid at each iteration \cite{tofighi2022decentralized,kim2019p2p,khorasany2019decentralized}. Second, the pricing scheme implemented by the main grid can be categorized into two approaches, enabling the establishment of distinct incentive prices among prosumers. This contributes to a greater level of fairness within the community grid.

To do that we first redefine (\ref{eq:Producers Cost Function}) and (\ref{eq:Consumer Utility Function}) to a welfare function. Let $p_{i0}$ denote the energy transfer from prosumer $i$ to the main grid (indexed by $0$ as a slack bus). Similarly, let $p_{ij}$ denote the energy transfer from prosumer $i$ to prosumer $j$. The energy exchanged during the P2P process of prosumer $i$ is expressed as
\begin{equation}
g_i-d_i=p_{i0}+\sum_{j\in \omega_i} p_{ij},
\label{eq:energy exchanged of prosumer}
\end{equation}
where $\omega_i$ denotes a set of prosumers to whom the prosumer $i$ sells. Then, the welfare function of a prosumer $i\in \mathcal{N}$ is defined as
\begin{align}
\begin{split}
W_i(d_i,g_i,p_{i0})&=U_i(d_i) - C_i(g_i)\\
&\ -\overline\lambda \max (p_{0i},0) +\underline\lambda \max (p_{i0},0), 
\label{eq:social welfare function of prosumer}
\end{split}
\end{align}
where $\bar\lambda$ is the buying price from the main grid, and $p_{0i}$ is the amount of energy from the main grid to the prosumer $i$. Note that $p_{0i}=-p_{i0}$ and $p_{ij}=-p_{ji}$. Similarly, $\underline\lambda$ is the selling price to the main grid.

To develop dynamic grid pricing, we factor out the cost function of main grid in (\ref{eq:main grid Cost Function}) as
\begin{equation}
C_0(p_0)=p_{0}(b_{0}+a_{0}p_{0}).
\label{eq:reformulate cost of pw}
\end{equation}
Then, we see that the term ($b_{0}+a_{0}p_{0}$) can serve as a price, which dynamically depends on $p_0$. Since $p_0=\sum_{i\in\mathcal{N}} p_{0i}$, we set the dynamic grid price $\overline\lambda$ as
\begin{equation}
\overline\lambda=b_0+a_0\sum_{i\in \mathcal{N}} p_{0i},
\label{eq:unique selling price}
\end{equation}
which implies that the buying price from the grid varies depending on the total community load $p_0$. Note that, however, all prosumers have the \textit{same} buying price in this case, and we call it \textbf{unique pricing scheme (UPS)}. The application of the UPS proposed in this paper and other relevant studies \cite{luo2021multiple,mehdinejad2022peer} may lead to potential market conflicts concerning the distribution of benefits between electricity consumers. Specifically, conflicts may arise between households characterized by low electricity consumption, and commercial/industrial consumers with high electricity consumption. This results in a situation where low-electricity consumers are compelled to pay higher prices due to the influence of high-consumption consumers and can seriously affect the long-term development of the community grid exchange market. 

Hence, we consider another pricing scheme called \textbf{differential pricing scheme (DPS)}, where each prosumer $i$ has it own buying price $\overline\lambda_{i}$ depending on its own load $p_{0i}$ such as
\begin{equation}
\overline\lambda_{i}=b_0+a_0 p_{0i}.
\label{eq:diff selling price}
\end{equation}
For realistic simulations, the selling price to main grid $\underline\lambda$ is fixed during a day \cite{eseye2017grid} or fixed at each time slot.
%new subsection
\subsection{Social Welfare Maximization of Community Grid} \label{subsec:SWM}
The objective of the community grid in \textbf{P1} can be reformulated as \textbf{P2} using (\ref{eq:1}) and (\ref{eq:social welfare function of prosumer}). Note that \textbf{P2} aims to maximize social welfare for community grid.

\vspace{1mm}
\noindent \textbf{P2: Social Welfare Maximization}
\begin{align}
\max& \sum_{i\in\mathcal{N}}  W_i(d_i,g_i,p_{i0}) \label{eq:SWM without penalty}\\
\text{s.t } & (\ref{eq:Producers energy limit}), (\ref{eq:Consumer energy limit}), (\ref{eq:energy exchanged of prosumer}) \text{  :virtual layer}, \nonumber \\
& (\ref{eq:1a}), (\ref{eq:1b}), (\ref{eq:1c}), (\ref{eq:capacity limit}), (\ref{eq:1f}), (\ref{eq:non-linear-constraint}) \text{  :physical layer}, \nonumber \\
\text{variables: } & \{p_{ij}, p_{i0}, d_{i}, g_{i}, v_{i}, l_{i}|\forall i\in\mathcal{N}\}. \nonumber
\end{align}
To solve the problem of \textbf{P2}, we consider the capacity of line in (\ref{eq:capacity limit}) can be exceeded but then a penalty is imposed in proportion to the mount of violation. Consequently, prosumers seek to avoid the trading on the congested lines, and \textbf{P2} is reformulated as \textbf{P3} congestion relaxed given a congestion price $\eta_{\ell,ij}$ as follow.

\vspace{1mm}
\noindent \textbf{P3: Social Welfare Maximization with Congestion Relaxation}
\begin{align}
\max& \sum_{i\in\mathcal{N}} \biggr[ W_i(d_i,g_i,p_{i0}) - \sum_{\ell\in \mathcal{L}}\sum_{j\in\omega_i} \eta_{\ell,ij} |p_{ij}|\biggr] \label{eq:SWM}\\
\text{s.t } & (\ref{eq:Producers energy limit}), (\ref{eq:Consumer energy limit}), (\ref{eq:energy exchanged of prosumer}) \text{  :virtual layer}, \nonumber \\
& (\ref{eq:1a}), (\ref{eq:1b}), (\ref{eq:1c}), (\ref{eq:1f}), (\ref{eq:non-linear-constraint}) \text{  :physical layer}, \nonumber \\
\text{variables: } &  \{p_{ij}, p_{i0}, d_{i}, g_{i}, v_{i}, l_{i}| \forall i\in \mathcal{N},j\in\omega_i\}. \nonumber
\end{align}
In the next section, we will discuss how to determine and update the congestion pricing $\eta_{\ell,ij}$ using the proposed two-stage model.
%%%%%% new section 
\section{Decentralized Two-stage Electricity Market}\label{sec:2-stage}
In solving \textbf{P3}, we aim to solve the local problem of each prosumer using only P2P communications to ensure the data privacy of prosumers. Therefore, we propose decentralized two-stage electricity market for efficient P2P energy trading and AC-OPF based on ADMM. Once P2P matching reaches an optimal solution in a virtual layer, the second stage for AC-OPF is initiated to verify whether physical constraints are satisfied. The benefit of this approach is that a central authority is completely avoided, and the data are only shared with the corresponding neighbors. Recall that our two-stage approach consists of P2P matching in the virtual layer and P2P realization in physical layer as shown in Fig.~\ref{fig:overview}.

%%%%%% new subsection 
\subsection{Decentralized P2P Matching in Virtual Layer} \label{subsec:P2P_negotiation}
In energy dispatch, the market seeks to minimize the total cost or maximize the total profit of prosumers. To this end, we exploit the modified decentralized optimization based on ADMM in \cite{baroche2019prosumer}. Accordingly, \textbf{P3} is decomposed into sub-problems in \textbf{P4} (virtual layer), which is iteratively solved by each prosumer $i$ and AC-OPF problem in \textbf{P5} (physical layer), which updates the congestion pricing at each round of interaction. As shown in Fig.~\ref{fig:2stage_algorithm}, \textbf{P3} is iteratively solved by \textbf{P4} and \textbf{P5}. Let $n$ be the index of iterative rounds between the virtual layer (\textbf{P4}) and the physical layer (\textbf{P5}).

At the $n^{th}$ round, given $\eta_{\ell,ij}=\eta_{\ell,ij}^n$, the sub-problem \textbf{P4} in virtual layer and its decentralized solution provided by prosumer $i$ is as follows:

\vspace{2mm}
\noindent \textbf{P4: P2P Mathing in Virtual Layer using modified ADMM}
\begin{align}
\begin{split}
(\{{p_{ij}\},p_{i0}})^{t+1}=\underset{\{p_{ij}\},p_{i0}} {\arg\min} -W_i(d_i,g_i,p_{i0}) + \frac{(p_{i0}-p_{i0}^t)^2}{\rho} \\ +\sum_{\ell\in \mathcal{L}}\sum_{j\in\omega_i} \eta_{\ell,ij}^n |p_{ij}| +\sum_{j\in\omega_i} \biggr[\frac{\rho}{2}\Bigl(\frac{p_{ij}^t-p_{ji}^t}{2} - p_{ij}+\frac{\lambda_{ij}^t}{\rho}\Bigl)^2\biggr], 
\end{split}\label{eq:P2P-ADMM} 
\end{align}
\begin{equation}
\text{s.t } (\ref{eq:Producers energy limit}), (\ref{eq:Consumer energy limit}), (\ref{eq:energy exchanged of prosumer}). \nonumber
\end{equation}
Note that, given $\eta_{\ell,ij}^n$, \textbf{P4}, which is identical to \textbf{P3} without physical constraints of (\ref{eq:1a})--(\ref{eq:1c}), (\ref{eq:1f}), (\ref{eq:non-linear-constraint}), is iteratively solvable. As done in \cite{baroche2019prosumer}, at $t^{th}$ iteration in the virtual layer, each prosumer $i\in\mathcal{N}$ finds the optimal quantities ${(\{p_{ij}\},p_{i0})}^{t+1}$ in (\ref{eq:P2P-ADMM}) depending on the quantities provided by its partners $p_{ji}^t$. Note that $\lambda_{ij}^{t}$ is the price of P2P trading, which is updated by the ADMM method such as \cite{baroche2019prosumer}
\begin{equation}
\lambda_{ij}^{t+1}=\lambda_{ij}^t-\frac{\rho}{2}(p_{ij}^{t+1}+p_{ji}^{t+1}). \label{eq:price of each bilateral}
\end{equation}
%%%%%%  new figure
\begin{figure}
\centering
\includegraphics[width=\columnwidth]{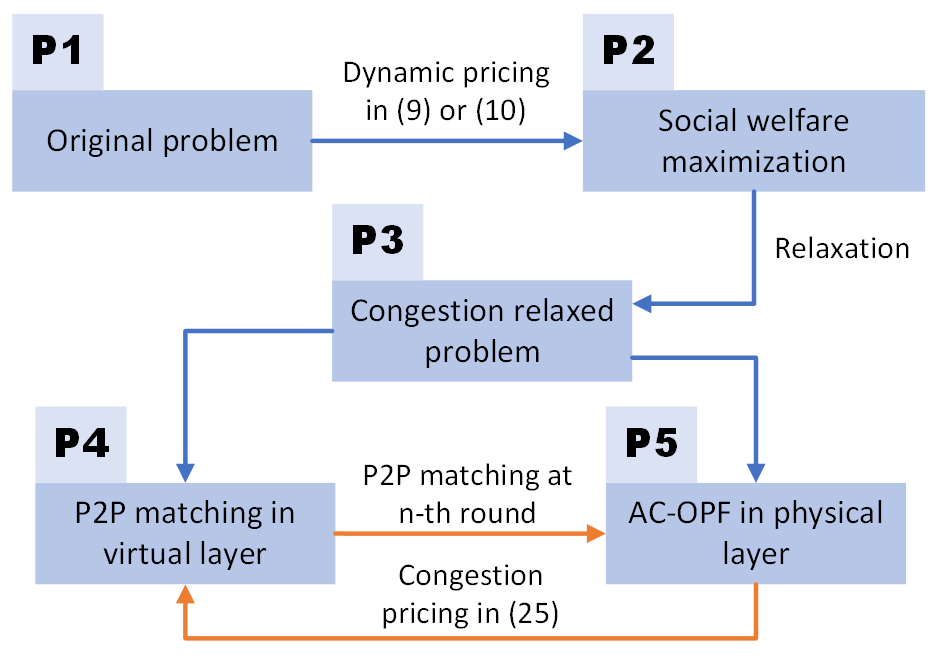}
\caption{Structure of problems in two-stage electricity market.}
\label{fig:2stage_algorithm}
\end{figure}
The convergence conditions are achieved by iteratively repeating the algorithm. Conditions (\ref{eq:primal value}) is evaluated as follows based on primal and dual residual values.
\begin{eqnarray}
||r_i^{t+1}||_2 \leq {\epsilon^2_{\text{pri}}}, ||s_i^{t+1}||_2 \leq {\epsilon^2_{\text{dual}}},\label{eq:primal value}
\end{eqnarray}
where $\epsilon_{\text{pri}}$ and $\epsilon_{\text{dual}}$ represent the primal and dual feasibility tolerances of a model, and their local primal and dual residuals are given by \cite{baroche2019prosumer}
\begin{eqnarray}
r_i^{t+1} =\sum_{j\in\omega_i} {(p_{ij}^{t+1}+p_{ji}^{t+1})}^2, s_i^{t+1} =\sum_{j\in\omega_i} {(p_{ij}^{t+1}-p_{ij}^t)}^2.\label{eq:local primal}
% s_i^{t+1} &=&\sum_{j=1} {(p_{ij}^{t+1}-p_{ij}^t)}^2 . \label{eq:dual residuals}
\end{eqnarray}
The P2P matching is summarized in \textbf{Algorithm 1}. Note that to improve the performance of ADMM, we adopt the step-size method of ADMM as described in \cite{boyd2011distributed} for both P2P matching and P2P realization.
\begin{algorithm}[t]
\caption{Decentralized P2P matching in virtual layer}\label{alg:cap}
\begin{flushleft}
        \parbox[t]{\dimexpr\linewidth-\algorithmicindent\relax}{%
        \setlength{\hangindent}{0pt}%
        \textbf{Input:} Given $\eta_{\ell,ij}^n, \forall \ell\in \mathcal{L}$ of \textbf{Algorithm 3};\\
        \textbf{Output:} $\{p_{ij}\}, p_{i0}, d_{i}, g_{i}$;\\
        \textbf{Initialization:} $t=0, p_{ij}^0=0, p_{i0}^0=0$, and $\lambda_{ij}^0=0, \forall i\in \mathcal{N},j\in\omega_i$;}\strut
\end{flushleft}
\begin{algorithmic}[1]
\While{(\ref{eq:primal value}) \textit{not satisfied}}
\State $t \leftarrow t+1$;
\State \parbox[t]{\dimexpr\linewidth-\algorithmicindent\relax}{%
    \setlength{\hangindent}{0pt}%
      Each prosumer $i$ solves (\ref{eq:P2P-ADMM}) and sends $p_{ij}$ to $j\in\omega_i$ and sends $p_{i0}$ to power grid;}\strut
\State \parbox[t]{\dimexpr\linewidth-\algorithmicindent\relax}{%
    \setlength{\hangindent}{0pt}%
      Grid receives $p_{i0}$ by each prosumer, it updates $\overline\lambda$ using UPS in (\ref{eq:unique selling price}) or DPS in (\ref{eq:diff selling price}) then sends $\overline\lambda$ back to prosumers;}\strut
\State \parbox[t]{\dimexpr\linewidth-\algorithmicindent\relax}{%
    \setlength{\hangindent}{0pt}%
      Each prosumer $i$ updates $\lambda_{ij}^{t}$ according to (\ref{eq:price of each bilateral}) and the local residuals in (\ref{eq:local primal});}\strut
\State \parbox[t]{\dimexpr\linewidth-\algorithmicindent\relax}{%
    \setlength{\hangindent}{0pt}%
      Each prosumer $i$ broadcasts its local residuals and receives the local residuals of it partners;\\
      \#Note that the residual in (\ref{eq:local primal}) is used to update ADMM step-size \cite{boyd2011distributed};}\strut
      \vspace{0.1em}
\EndWhile
\end{algorithmic}
\end{algorithm}

\subsection{Decentralized P2P Realization in Physical Layer} \label{subsec:P2P_realization}
In solving \textbf{P3}, we first solved \textbf{P4} in a virtual layer. But the congestion price $\eta_{\ell,ij}^n$ should be determined, and thus we need to solve an OPF problem in physical layer. To avoid sharing network parameters, we consider a decentralized OPF, where the network is decomposed into multiple zones \cite{dvorkin2020differentially}. To enable distributed computation, squared voltage angles $v_i$ are duplicated per zone, which is defined as $\vartheta_i\in\mathbb{R}^{|\mathcal{N}|}$ and updated across iterations to satisfy the following consensus constraint:
\begin{equation}
\vartheta_i=\overline \vartheta \colon \mu_i, \forall i\in\mathcal{N}, \label{eq:enforce consensus constraint}
\end{equation}
where $\overline \vartheta\in\mathbb{R}^{|\mathcal{N}|}$ represents the consensus variable and $\mu_i\in\mathbb{R}^{|\mathcal{N}|}$ indicates the dual variable corresponding to the consensus constraint of each node $i$. Note that each node $i$ has its own zone as shown in Fig.~\ref{fig:de_ac-opf}.

The decomposition in (\ref{eq:enforce consensus constraint}) separates the feasibility regions for \textbf{P3} in physical layer per zone; therefore, we denote the constraint set of each zone by $\mathcal{F}_i$. Then, it is necessary to optimize the following partial Lagrangian function to solve OPF per zone. In doing this, we define $\mu=\{\mu_i|i\in\mathcal{N}\}$, $\hat{g}=\{\hat{g}_i|i\in\mathcal{N}\}$ where $\hat{g_i}=g_i+g_{i,\text{loss}}$, $\ell=\{\ell_i|i\in\mathcal{N}\}$, $\vartheta=\{\vartheta_i|i\in\mathcal{N}\}$, which are collection of variables of all nodes. Note that we introduce a new variable $g_{i,\text{loss}}$, which denotes \textit{additional} power generation to compensate the power losses. Since $g_i$ in \textbf{P2}, \textbf{P3}, \textbf{P4} is the power generation in virtual layer, power losses cannot be captured. Hence, in physical layer optimization, we consider the cost incurred for power losses compensation given by $C_{i,\text{loss}}(g_{i,\text{loss}}^P)\triangleq C_i(\hat{g_{i}}^P)-C_i(g_{i}^P)$. Since this term cannot be considered for P2P transaction, we assure that cost incurred from power losses compensation can be covered from a monthly fee as done in \cite{SonnComm}.

\vspace{2mm}
\noindent \textbf{P5: AC-OPF in Physical Layer using ADMM}
\begin{equation}
\underset{\mu}{\max}\underset{\hat{g},l,\vartheta,\overline{\vartheta}}{\min}\mathcal{L}(\hat{g},l,\vartheta,\overline{\vartheta},\mu) =\sum_{i\in\mathcal{N}} \biggr[ C_{i,\text{loss}}(g_{i,\text{loss}})+\mu_i^\top(\overline{\vartheta}-\vartheta_i) \biggr],
\label{eq:lagrangian_function}
\end{equation}
\begin{align}
\text{s.t.} & \quad (\ref{eq:1a}), (\ref{eq:1b}), (\ref{eq:1c}), (\ref{eq:non-linear-constraint}), (\ref{eq:1f}) \in\cap_{i\in\mathcal{N}}\mathcal{F}_i. \nonumber
\end{align}

Distributed OPF computation is given by the following ADMM algorithm \cite{dvorkin2020differentially}. 
\begin{align}
\begin{split}
\vartheta_i^{t+1} =\underset{\hat{g_{i}},l_i,\vartheta_i}{\arg\min}\mathcal{L}_i(\hat{g_{i}},l_i,\vartheta_i,\overline{\vartheta}^t,\mu_{i}^t)  +\frac{\rho}{2}||\overline{\vartheta}^t-\vartheta_i||_2^2,\\
\end{split} \label{eq:squared_voltage_admm}
\end{align}
\begin{equation}
\overline{\vartheta}^{t+1} =\underset{\overline{\vartheta}} {\arg\min}\mathcal{L}(\hat{g},l,\vartheta^{t+1},\overline{\vartheta},\mu^t) +\frac{\rho}{2}\sum_{i\in\mathcal{N}} ||\overline{\vartheta}-\vartheta_i^{t+1}||_2^2, \label{eq:vector_consensus_admm}
\end{equation}
\begin{equation}
\mu_i^{t+1} =\mu_i^t+\rho(\overline{\vartheta}^{t+1}-\vartheta_i^{t+1}). \label{eq:dual-coordination}
\end{equation}
A round index $t$ along with a penalty factor $\rho$ is used to represent the ADMM regularization term. The process is repeated until the termination condition is reached. Condition (\ref{eq:gap_squared_voltage}) is used to evaluate the gap in the values between two iterations of squared voltage angles, and defined as follows:
\begin{equation}
\sum_{i\in\mathcal{N}} ||\overline{\vartheta}^{t+1}-{\vartheta}_i^{t+1}||_2 \leq\epsilon. \label{eq:gap_squared_voltage}
\end{equation}
The primal and dual residual are defined as 
\begin{equation}
r_i^{t+1}=\sum_{i\in\mathcal{N}} ||\overline{\vartheta}^{t+1}-\vartheta_i^{t+1}||_2, s_i^{t+1}=\sum_{i\in\mathcal{N}} ||\vartheta_i^{t+1}-\vartheta_i^t||_2.\label{eq:dual_ac-opf}
\end{equation}
% \begin{equation}
% r^{t+1}=\sum_i^N ||\overline{v}_i^{t+1}-v_i^{t+1}||_2. \label{eq:primal_ac-opf}
% \end{equation}
%new subsubsection
The process of solving \textbf{P5} is summarized in \textbf{Algorithm 2}.

\begin{figure}[t]
\centering
\includegraphics[width=\columnwidth]{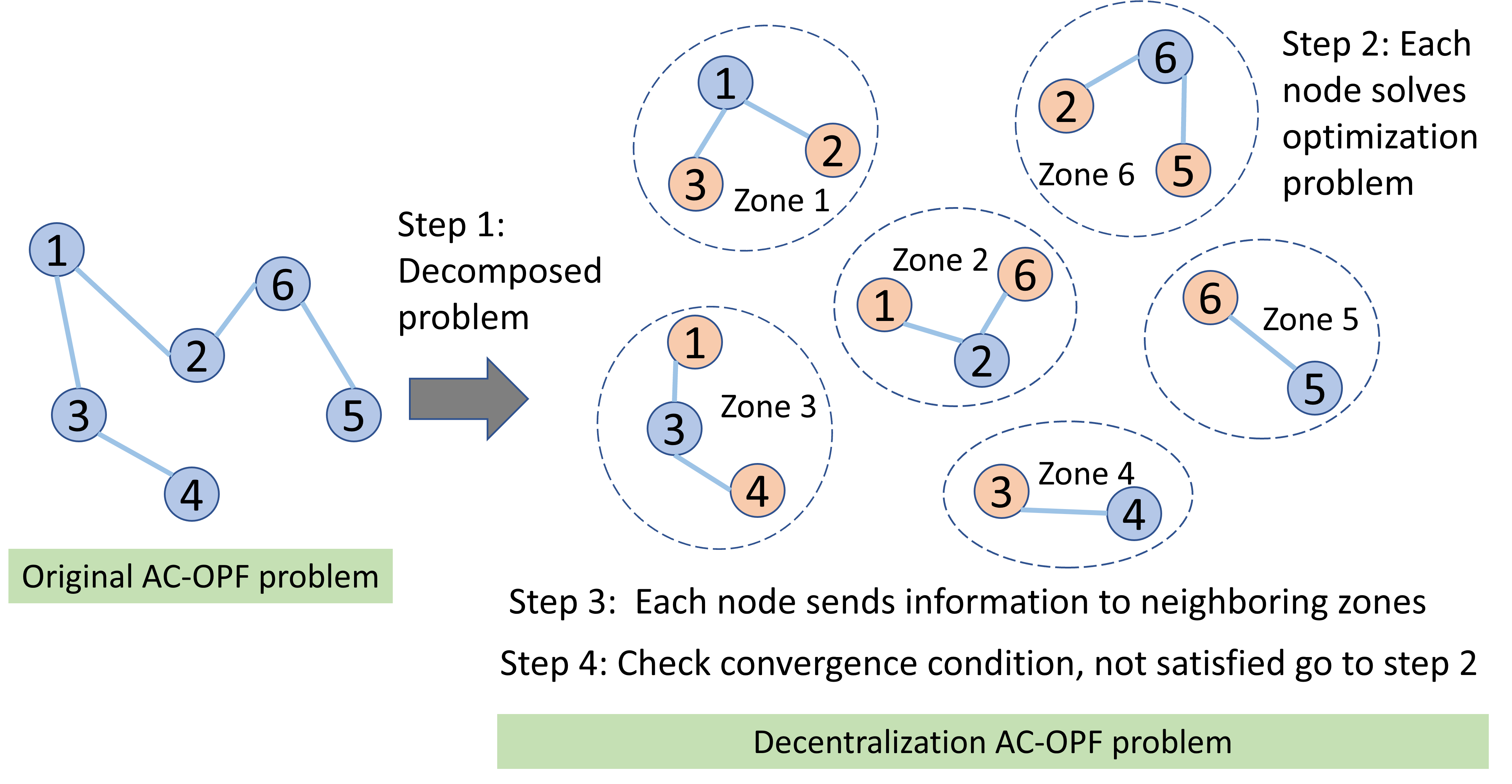}
\caption{Decentralized AC-OPF problem.}
\label{fig:de_ac-opf}
\end{figure}

\begin{algorithm}
\caption{Decentralized P2P realization in physical layer}\label{alg:realization}
\begin{flushleft}
        \textbf{Input:} Given $p_{ij}, p_{0i}, d_{i}, g_{i}$ of \textbf{Algorithm 1};\\
        \textbf{Output:} $g_{i,\text{loss}},l_i,\vartheta_i,\overline{\vartheta}, f_i^P, f_i^Q$;\\
        \textbf{Initialization:} $t=0, g_{i,\text{loss}}=0,l_i=0, \vartheta_i^0=0, \overline{\vartheta}^0=0$, and $\mu_i^{0}=0, \forall i\in \mathcal{N}$;
\end{flushleft}
\begin{algorithmic}[1]
\While{(\ref{eq:gap_squared_voltage}) \textit{not satisfied} }
\State $t \leftarrow t+1$;
\State \parbox[t]{\dimexpr\linewidth-\algorithmicindent\relax}{%
    \setlength{\itemindent}{0em}%
      Each node $i$ solves the local problem in (\ref{eq:squared_voltage_admm}) and sends the result to other nodes in zone;}\strut 
\State \parbox[t]{\dimexpr\linewidth-\algorithmicindent\relax}{%
    \setlength{\itemindent}{0em}%
      Each node $i$ updates consensus following (\ref{eq:vector_consensus_admm}) and dual-coordination signals following (\ref{eq:dual-coordination});}\strut
\State \parbox[t]{\dimexpr\linewidth-\algorithmicindent\relax}{%
    \setlength{\itemindent}{0em}%
      Each node $i$ broadcasts the updated consensus and its residual in (\ref{eq:dual_ac-opf}) to another node; \\
      \#Note that the residual in (\ref{eq:dual_ac-opf}) is used to update ADMM step-size \cite{boyd2011distributed};}\strut
      \vspace{0.1em}
\EndWhile
\end{algorithmic}
\end{algorithm}
%%%%%%  new subsubsection 
\subsection{Congestion Price} \label{subsec:congestion_rule}
Finally, we present the congestion price update, which serves as the \textit{corner}store of our two-stage algorithm. When congestion occurs on line $\ell \in \mathcal{L}$, P2P participants must adjust their P2P matching until they are adapted to fit the power system. The proposed congestion pricing encourages each node to reduce the amount of P2P exchanged energy in the congested line $\ell$. Congestion price is imposed on the overloading of energy when prosumers inject more energy beyond the line capacity. 

At the $n^{th}$ round, let $\kappa_{\ell}^n$ denote the overloading power calculated based on the deviation rate between the active power of P2P matching and the power flow of AC-OPF (active, reactive, and compensation power losses) in line $\ell$ as shown in (\ref{eq:overloading amount}). 
\begin{equation}
\kappa_{\ell}^n=\left(\frac{S_{\ell}^n}{S^{\max}_{\ell}}-1\right) \sum_{(i,j)\in \mathcal{M}_\ell} p_{ij}^n. \label{eq:overloading amount}
\end{equation}
The first component presents the deviation rate between the power flow $S_{\ell}^n$ and the maximum capacity $S^{\max}_{\ell}$ of line $\ell$ at the $n^{th}$ round. The second component is the total amount of active power exchanged over line $\ell$ at the $n^{th}$ round where $\mathcal{M}_\ell$ is a set of pairs $(i,j)$ that use line $\ell$ to transfer energy.

Note that here we focus on active power only but reactive power and losses in lines are affected by active power. Hence, a change in active power results in a change in reactive power and losses in lines as well. Thus, the congestion pricing $\eta_{\ell,ij}^n$ is updated as follows
\begin{equation}
\eta_{\ell,ij}^n=\eta_{\ell,ij}^{n-1}+\gamma\kappa_{\ell}^{n},\label{eq:congestion cost}
\end{equation}
where $\gamma$ denotes the penalty parameter set by the main grid. The proposed two-stage algorithm using congestion pricing is summarized in \textbf{Algorithm 3}.
\begin{algorithm}
\caption{Decentralized two-stage electricity market}\label{alg:two-stage}
\begin{flushleft}
        \parbox[t]{\dimexpr\linewidth-\algorithmicindent\relax}{%
        \setlength{\hangindent}{0pt}%
        \textbf{Input:} Prosumers request P2P market participation;\\
        \textbf{Output:} Market clearing solution;\\
        \textbf{Initialization:} $\overline\lambda$, $\underline\lambda$, $j\in\omega_i$ , $n=1$, and $\eta_{\ell,ij}^1=0$;}\strut
\end{flushleft}
\begin{algorithmic}[1]
\While{true}
\State \parbox[t]{\dimexpr\linewidth-\algorithmicindent\relax}{%
    \setlength{\hangindent}{0pt}%
      Execute \textbf{Algorithm~\ref{alg:cap}};}\strut
\State \parbox[t]{\dimexpr\linewidth-\algorithmicindent\relax}{%
    \setlength{\hangindent}{0pt}%
      Execute \textbf{Algorithm~\ref{alg:realization}};}\strut
\State \parbox[t]{\dimexpr\linewidth-\algorithmicindent\relax}{%
    \setlength{\hangindent}{0pt}%
      Each node examines congestion independently;}\strut
\If{congestion detected} 
    \State \parbox[t]{\dimexpr\linewidth-\algorithmicindent\relax}{%
    \setlength{\hangindent}{0pt}%
      Updates congestion price based on (\ref{eq:overloading amount}) and (\ref{eq:congestion cost});}\strut
\Else{}
    \State Stop two-stage market;
\EndIf 
\State \parbox[t]{\dimexpr\linewidth-\algorithmicindent\relax}{%
    \setlength{\hangindent}{0pt}%
      Nodes notify prosumers to update congestion price;}\strut
\State $n \leftarrow n+1$;
\EndWhile
\end{algorithmic}
\end{algorithm}

\begin{proposition} \label{prop:1} (convergence to an optimal solution) The proposed decentralized two-stage electricity market algorithm using the congestion pricing in (\ref{eq:congestion cost}) converses to an optimal solution of \textbf{P2}. 
\end{proposition}
\begin{proof} The proof is based on the network utility maximization \cite{low1999optimization}, but the detailed proof is omitted due to space limitation. 
\end{proof}
%new section
\section{Performance Evaluation}\label{sec:performance}
\subsection{Simulation setup}
We numerically evaluate the feasibility of the designed market and the effectiveness of the proposed approach in clearing the market. For detailed examination, we first consider a small network with a 15-bus low voltage radial grid from \cite{mieth2019distribution} as illustrated in Fig.~\ref{fig:network_simulation}, which includes the network parameters obtained from \cite{dvorkin2020distribution}. The coefficients for consumer demand and producers of DERs are from \cite{dvorkin2020distribution}. All prosumers are listed in Table~\ref{tab:con_pars} and Table~\ref{tab:pro_pars}, and the line parameters are presented in Table~\ref{tab:line_pars}. The coefficients for main grid as a producer is $a_0=1$\textdollar$/MWh^2$ and $b_0=25 $\textdollar$/MWh$. The congestion parameter $\gamma=0.5$. For ADMM, the change of penalties $\rho$ for both P2P matching and P2P realization are set from $10^{-4}$ to $10^5$. In addition, $\epsilon_{\text{pri}}=\epsilon_{\text{dual}}=10^{-6}$ in P2P matching while $\epsilon=10^{-4}$ for P2P realization. Simulating is done using Gurobi \cite{Gurobioptimization}.
\begin{figure}
\centering
\includegraphics[width=\columnwidth]{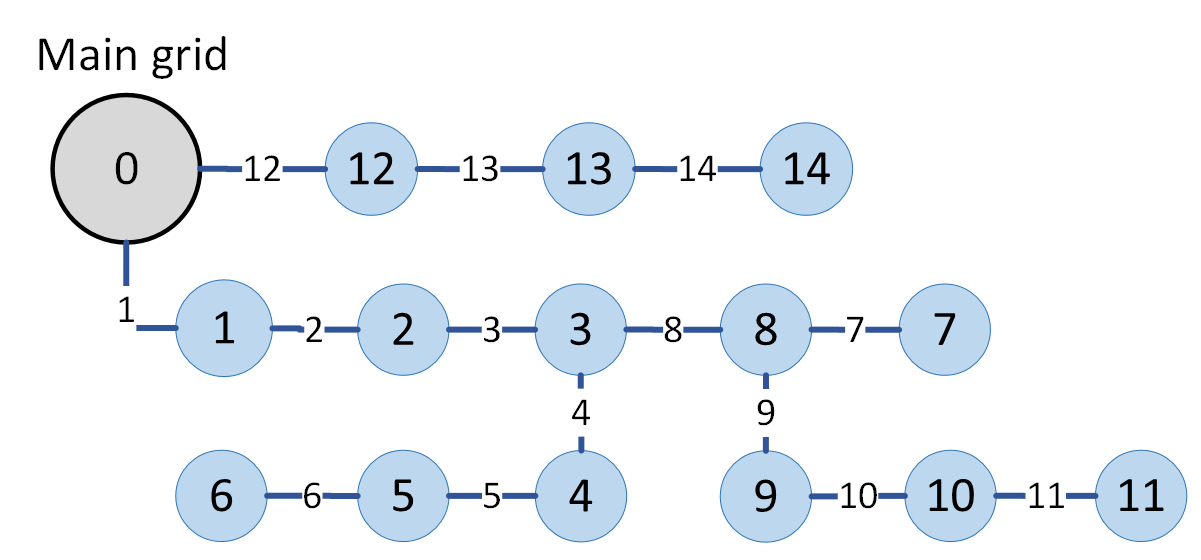}
\caption{Low-voltage 15-bus radial network for evaluation.}
\label{fig:network_simulation}
\end{figure}

\begin{table}[t]
\vspace{0.5em}
\caption{Consumer parameters.}
\label{tab:con_pars}
\resizebox{0.6\columnwidth}{!}{%
\begin{tabular}{|c|c|c|c|c|}
\hline
Node $i$ & $d_i^P$  & $d_i^Q$  & $\alpha_i$  & \ $\beta_i$   \\ \hline
0    & 1000   & 1000   & ---  & ---   \\ \hline
1    & 0.0201 & 0.0068 & 0.99 & 30.12 \\ \hline
2    & 0.0201 & 0.0084 & 0.78 & 14.11 \\ \hline
3    & 0.0201 & 0.0084 & 0.7  & 19.13 \\ \hline
4    & 0.173  & 0.053  & 0.95 & 25.67 \\ \hline
5    & 0.0291 & 0.0073 & 0.86 & 11.25 \\ \hline
6    & 0.219  & 0.12  & 0.76 & 40.28 \\ \hline
7    & 0.0235 & 0.0033 & 0.69 & 12.22 \\ \hline
8    & 0.0235 & 0.0059 & 0.58 & 35.08 \\ \hline
9    & 0.229  & 0.19   & 0.92 & 11.03 \\ \hline
10   & 0.0217 & 0.0065 & 0.53 & 17.06 \\ \hline
11   & 0.0132 & 0.0033 & 0.7  & 25.16 \\ \hline
12   & 0.6219 & 0.2951 & 0.83 & 25.6  \\ \hline
13   & 0.0224 & 0.0083 & 0.69 & 12.02 \\ \hline
14   & 0.0224 & 0.0083 & 0.75 & 17.68 \\ \hline
\end{tabular}}
\centering
\vspace{0.5em}
\end{table}
\begin{table}[t]
\caption{Producer parameters.}
\label{tab:pro_pars}
\resizebox{0.7\columnwidth}{!}{%
\begin{tabular}{|c|c|c|c|c|c|c|}
\hline
Node $i$ & $g_i^P$  & $g_i^Q$  & $v_i^{\min}$ & $v_i^{\max}$   & $a_i$    & $b_i$  \\ \hline
0    & 1000 & 1000 & 1.21 & 0.81 & ---  & ---   \\ \hline
1    & 0.4  & 0.2  & 1.21 & 0.81 & 0.57 & 9.34  \\ \hline
2    & 0.4  & 0.2  & 1.21 & 0.81 & 0.84 & 8.49  \\ \hline
3    & 0.4  & 0.2  & 1.21 & 0.81 & 0.96 & 14.3  \\ \hline
4    & 0.4  & 0.2  & 1.21 & 0.81 & 0.55 & 17.64 \\ \hline
5    & 0.4  & 0.2  & 1.21 & 0.81 & 0.89 & 15.46 \\ \hline
6    & 0.4  & 0.2  & 1.21 & 0.81 & 0.87 & 17.48 \\ \hline
7    & 0.4  & 0.2  & 1.21 & 0.81 & 0.62 & 11.38 \\ \hline
8    & 0.4  & 0.2  & 1.21 & 0.81 & 0.76 & 9.55  \\ \hline
9    & 0.4  & 0.2  & 1.21 & 0.81 & 0.99 & 19.17 \\ \hline
10   & 0.4  & 0.2  & 1.21 & 0.81 & 0.79 & 18.21 \\ \hline
11   & 0.4  & 0.2  & 1.21 & 0.81 & 0.68 & 19.89 \\ \hline
12   & 0.4  & 0.2  & 1.21 & 0.81 & 0.92 & 7.34  \\ \hline
13   & 0.4  & 0.2  & 1.21 & 0.81 & 0.58 & 14.47 \\ \hline
14   & 0.4  & 0.2  & 1.21 & 0.81 & 0.8  & 18.47 \\ \hline
\end{tabular}}
\centering
\end{table}
\begin{table}[t]
\vspace{0.5em}
\caption{Line parameters.}
\label{tab:line_pars}
\centering
\resizebox{0.7\columnwidth}{!}{%
\begin{tabular}{|c|c|c|c|c|c|}
\hline
Line $i(=\ell)$ & From & To & $r_i$    & $x_i$    & $S_\ell^{\max}$    \\ \hline
1   & 0    & 1  & 0.001  & 0.12   & 2     \\ \hline
2   & 1    & 2  & 0.0883 & 0.1262 & 0.512 \\ \hline
3   & 2    & 3  & 0.1384 & 0.1978 & 0.512 \\ \hline
4   & 3    & 4  & 0.0191 & 0.0273 & 0.256 \\ \hline
5   & 4    & 5  & 0.0175 & 0.0251 & 0.256 \\ \hline
6   & 5    & 6  & 0.0482 & 0.0689 & 0.256 \\ \hline
7   & 8    & 7  & 0.0523 & 0.0747 & 0.256 \\ \hline
8   & 3    & 8  & 0.0407 & 0.0582 & 0.256 \\ \hline
9   & 8    & 9  & 0.01   & 0.0143 & 0.256 \\ \hline
10  & 9    & 10 & 0.0241 & 0.0345 & 0.256 \\ \hline
11  & 10   & 11 & 0.0103 & 0.0148 & 0.256 \\ \hline
12  & 0    & 12 & 0.001  & 0.12   & 1     \\ \hline
13  & 12   & 13 & 0.1559 & 0.1119 & 0.204 \\ \hline
14  & 13   & 14 & 0.0953 & 0.0684 & 0.204 \\ \hline
\end{tabular}}
\end{table}

%%%% new subsection
\subsection{Decentralized OPF and congestion management}
Now we present the simulation results and illustrate voltage variation and losses in lines to deal with congestion.
%new subsubsection
\subsubsection{Decentralized OPF} The goal of AC-OPF is to minimize the cost of power losses owing to the transfer of energy in lines. Fig.~\ref{fig:compare_voltage} illustrates the voltages at all the nodes. The voltage magnitudes range from 0.96 pu to 1.03 pu at each node. The performance of the proposed decentralized method is same to that of the centralized one, and the total losses are all 0.0121 $MWh$. Consequently, the decentralized AC-OPF achieves the optimal solution while minimizing information exchange in a distribution network, reducing the risk of confidential information leaks.
\begin{figure}
\centering
\includegraphics[width=\columnwidth]{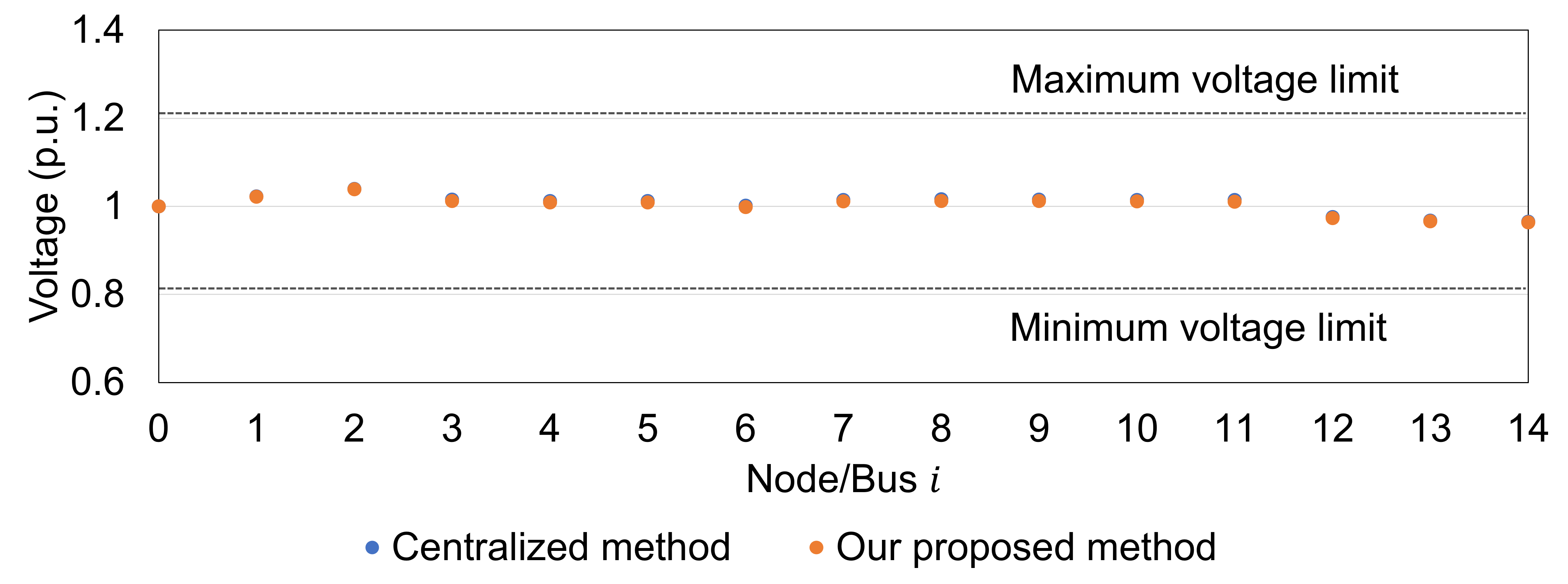}
\caption{Voltage magnitude at each node in 15-bus.}
\label{fig:compare_voltage}
\end{figure}
\begin{figure}
\centering
\includegraphics[width=\columnwidth]{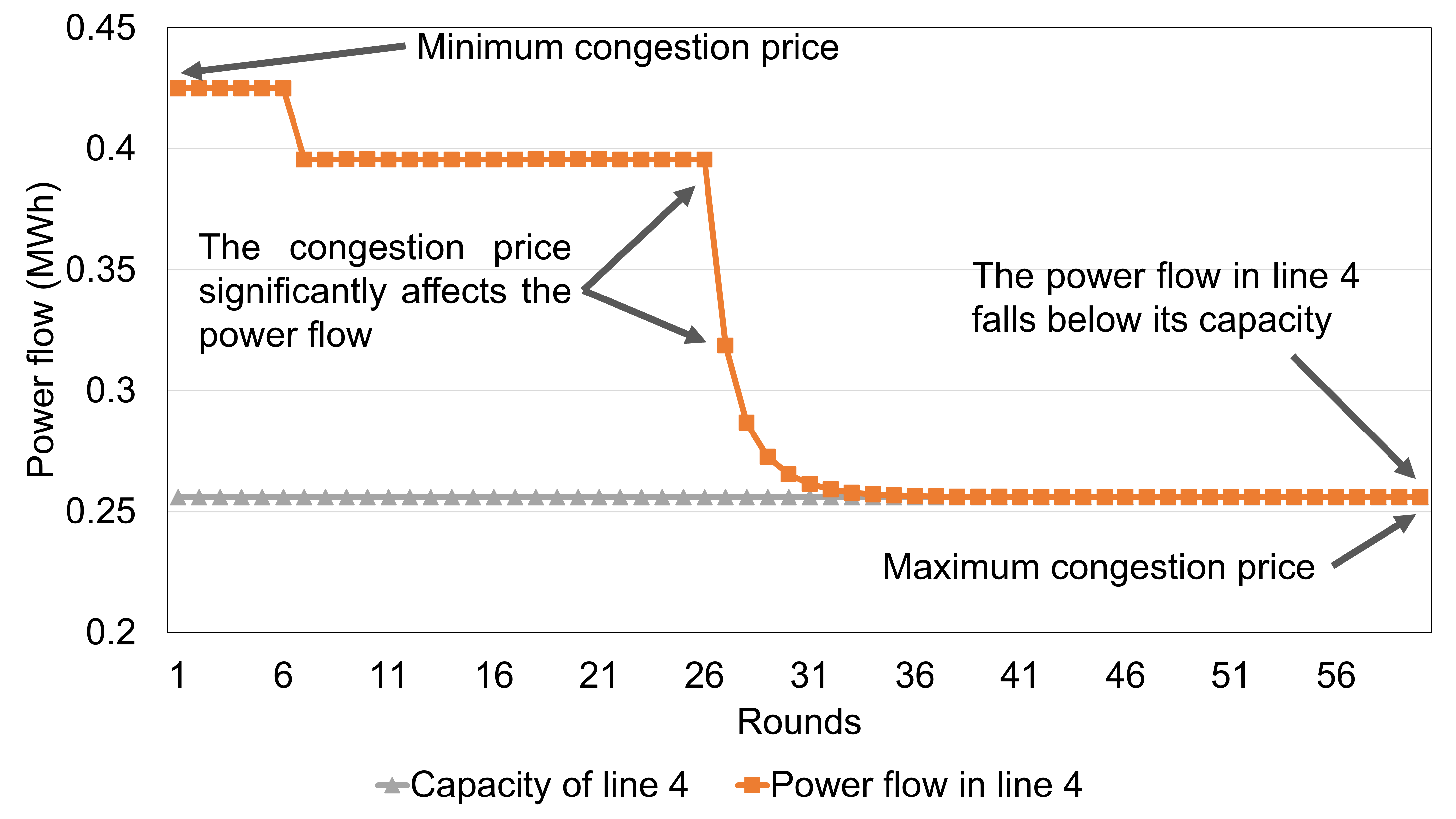}
\caption{Power flow of Line 4 in all rounds in 15-bus.}
\label{fig:Power flow of line 4}
\end{figure}
\subsubsection{Congestion Pricing} After solving AC-OPF, overloading of the network is determined, and congestion pricing is updated at each node. In this experiment, congestion occurs on Line 4. The change in the power flow of Line 4 during the operation of the P2P solution and AC-OPF is depicted in Fig.~\ref{fig:Power flow of line 4}. The first round is from the start of the algorithm where the first detection of congestion at nodes. The last round indicates the point at which no congestion occurs. Our method converges in 60 rounds, and the electricity market reaches equilibrium after 94.8 seconds.
\begin{figure}
\vspace{0.5em}
\centering
\includegraphics[width=\columnwidth]{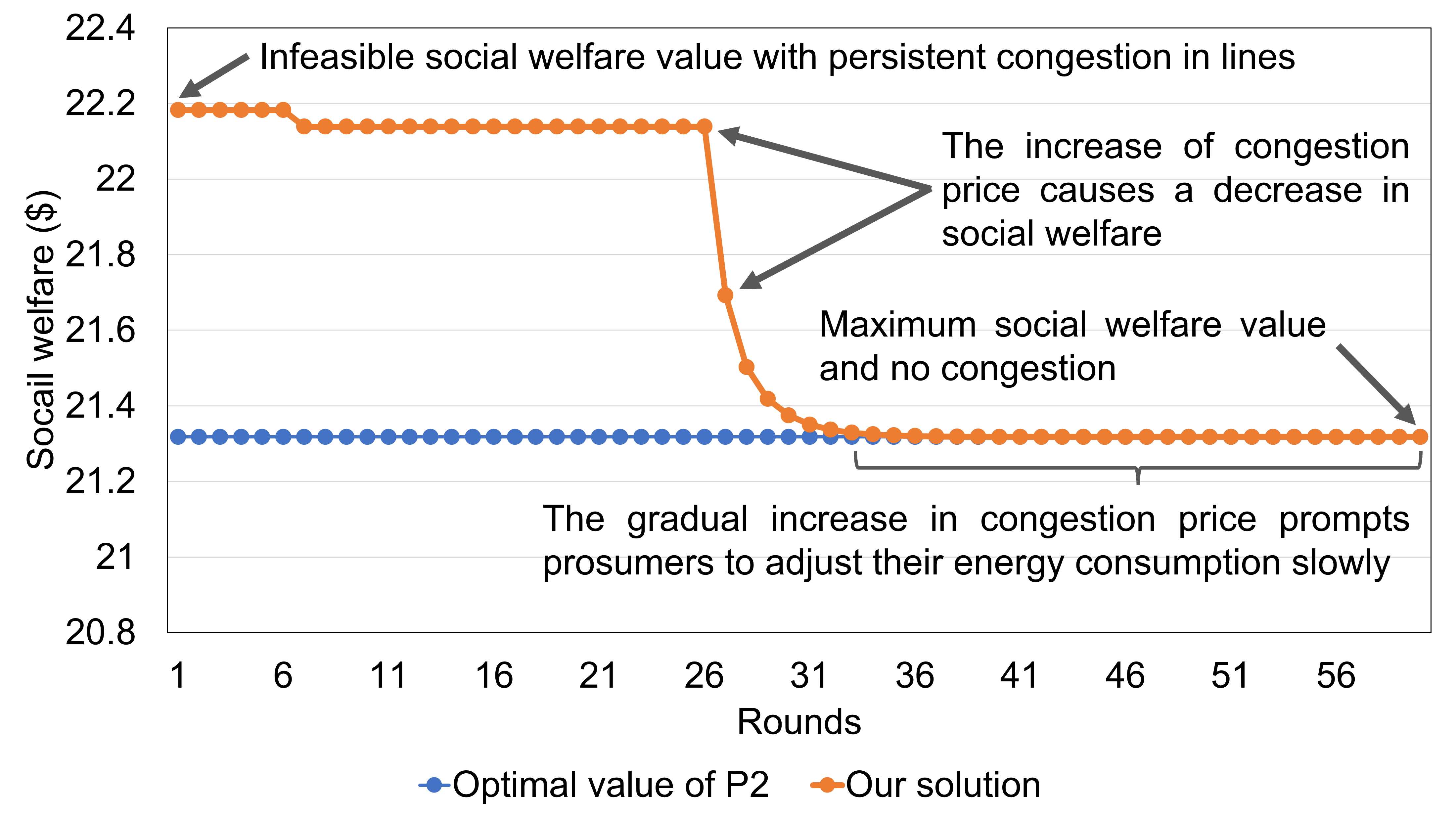}
 \caption{Social welfare in case 15-bus following rounds.}
\label{fig:social_welfare_rounds}
\end{figure}
\begin{figure}
\centering
\includegraphics[width=\columnwidth]{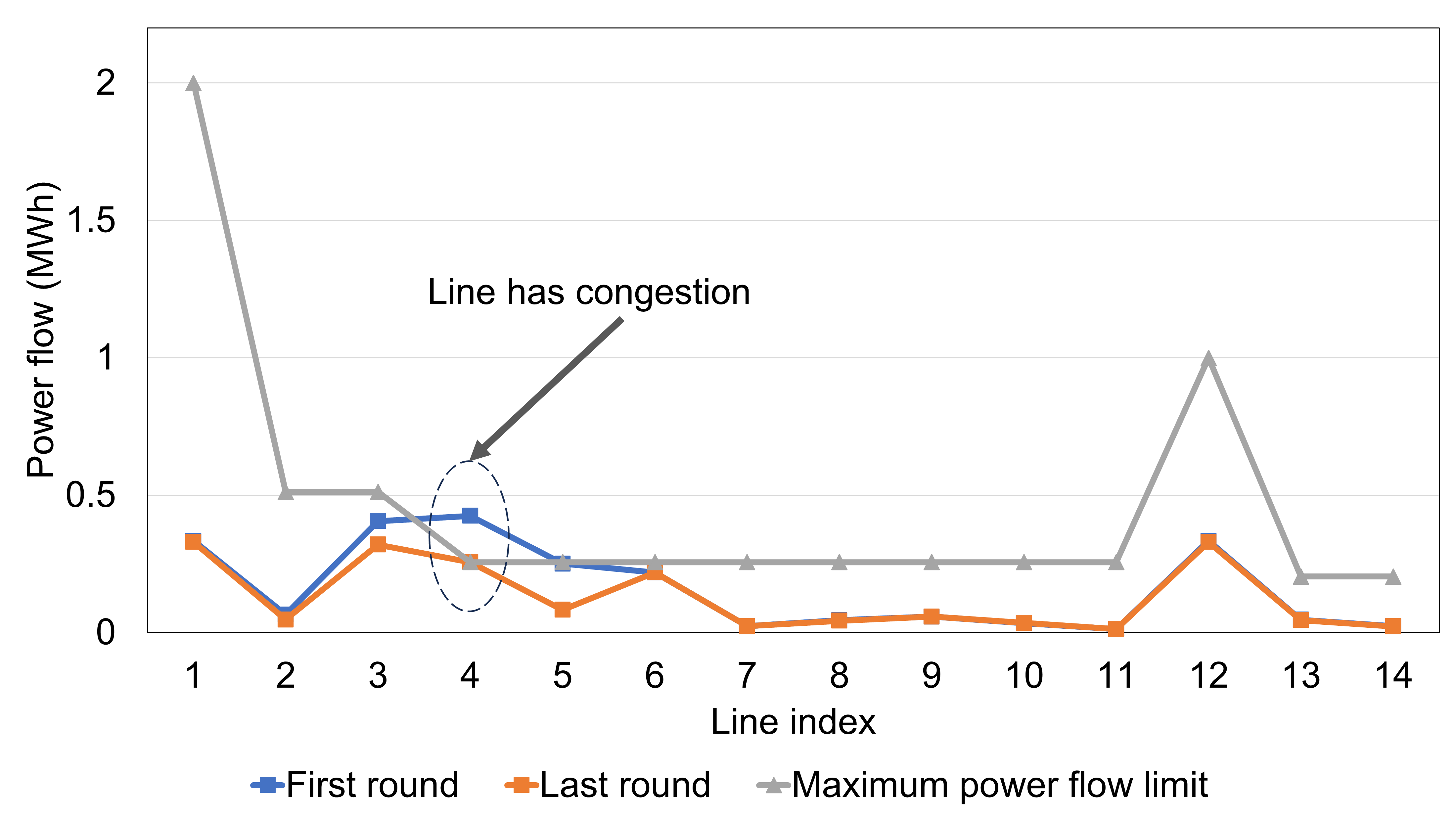}
 \caption{Power flow of all lines in 15-bus.}
\label{fig:active_all_lines}
\end{figure}
\begin{figure}
\centering
\includegraphics[width=\columnwidth]{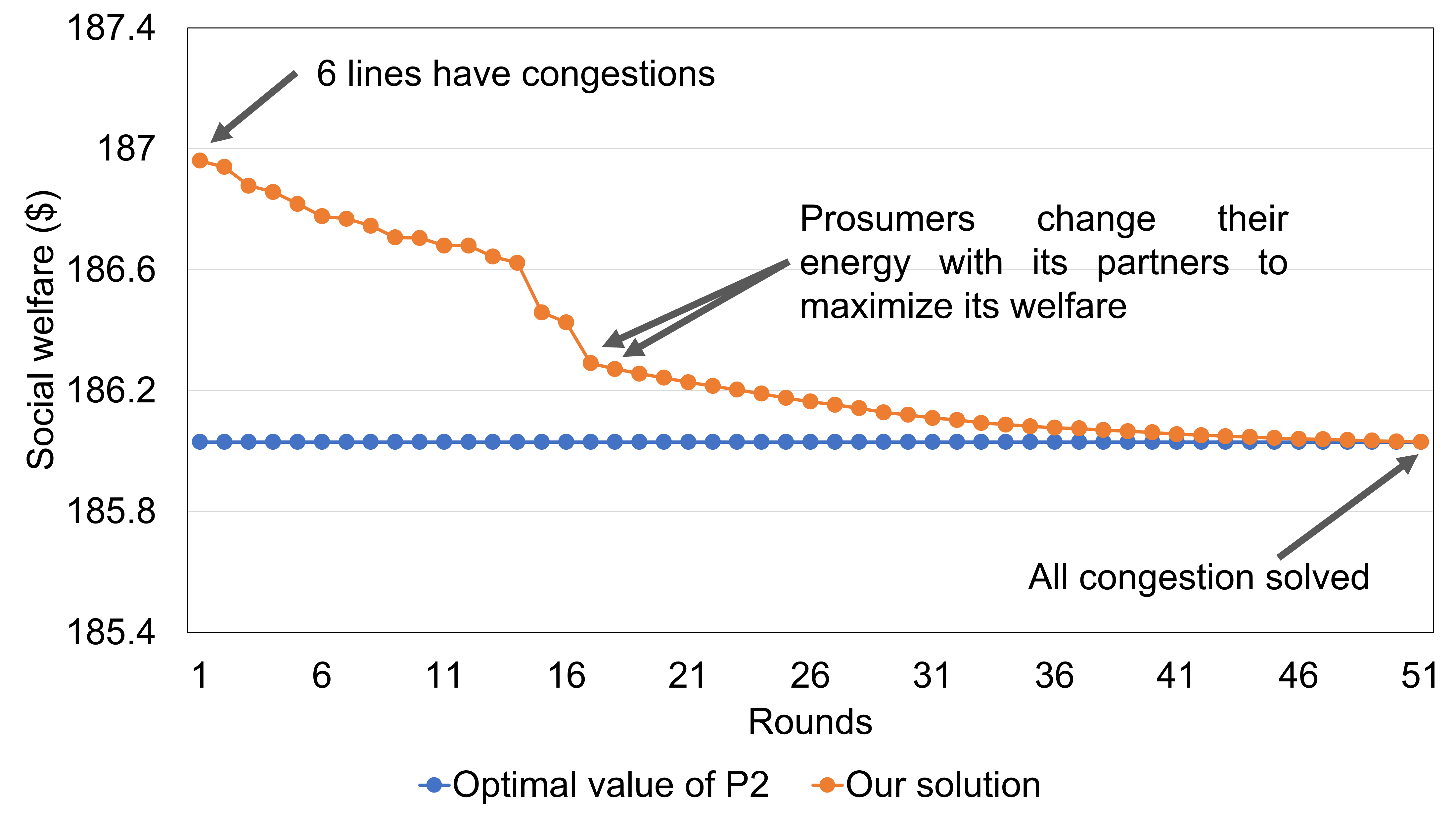}
 \caption{Social welfare in case 141-bus following rounds.}
\label{fig:social_welfare_141_rounds}
\end{figure}

As mentioned in (\ref{eq:congestion cost}), congestion price is determined from the power flow in the congested line, and the previous congestion price. As can be seen in Fig.~\ref{fig:Power flow of line 4}, the power flow falls below the line capacity only after 60 rounds of iterations. It results in prosumers reducing their amounts in matched pairs and encouraging them to trade with another prosumers instead. One may wonder why the power flow curve appears to be the same, for example, from 7 to 26 in Fig.~\ref{fig:Power flow of line 4}. The reason is as follows. The benefits derived from trading on another lines are lower than those derived from trading on the congested line $\ell$. Therefore, prosumers maintains a preference for trades on $\ell$, even if it entails incurring additional costs. Thus, the power flow curve may remain the same even if a penalty is applied each round. Besides, the increased congestion price leads to change energy $p_{ij}$, resulting in the change in social welfare, and converges to an optimal social welfare of \textbf{P2}, as shown in Fig.~\ref{fig:social_welfare_rounds}. The final result of congestion for all lines is shown in Fig.~\ref{fig:active_all_lines}.
\begin{table}[t]
\vspace{0.5em}
\caption{Comparison of average price between first and last rounds.}
\label{tab:compre_average_price}
\centering
\resizebox{0.8\columnwidth}{!}{%
\begin{tabular}{|c|c|c|}
\hline
                                                 & First round & Last round \\ \hline
\makecell{Average P2P price\\of prosumers using \\congested Line 4} & 9.71 (\textdollar$/MWh$)             & 12.66 (\textdollar$/MWh$)         \\ \hline
\makecell{Average P2P price\\of all prosumers}                 & 9.71 (\textdollar$/MWh$)             & 10.2 (\textdollar$/MWh$)         \\ \hline
\end{tabular}}
\vspace{0.5em}
\end{table}
%new table
\begin{table}
% \vspace{0.5em}
\caption{Matched pairs in first and last rounds in Line 4.}
\label{tab:matched_pairs}
\centering
\resizebox{0.8\columnwidth}{!}{%
\begin{tabular}{|c|c|c|}
\hline
                                                      & First round & Last round \\ \hline
\multirow{3}{*}{\makecell{Matched pairs using \\congested Line 4}} & (1,4); (1,5);  (1,6);   & (1,4); (1,6);   \\ 
                                                      &  (2,4); (2,5); (2,6);    & (2,4); (2,6);    \\  
                                                      &  (8,4); (8,5); (8,6);   & (12,4); (12,6);    \\  
                                                      &  (12,4); (12,5); (12,6);   &  \\  \hline
\end{tabular}}
\end{table}

Table~\ref{tab:compre_average_price} indicates that using the proposed technique, the average market clearing price of all consumers converges to 9.71\textdollar/$MWh$ and congestion increases the average price significantly. Specifically, at the first round, when congestion is not taken into account, the average price on Line 4 is 9.71\textdollar/$MWh$. At the last round, however, the average clearing price of pairs on congested Line 4 is increased to 12.66\textdollar/$MWh$. It lead to the total clearing price of the P2P market changes to 10.2\textdollar/$MWh$. Furthermore, according to the obtained results, the congestion not only affects the average price but also affects the matching pairs that use the congested line. The results are presented in Table~\ref{tab:matched_pairs}, where the number of matched pairs has decreased from 12 to 6. Based on the results, the proposed scheme could potentially relieve the grid congestion while improving P2P energy management between prosumers. The detail of clearing prices and quantities can be seen in Table~\ref{tab:trade_prices} and Table~\ref{tab:trade_quantities}, respectively. 

To check the performance in a larger system, P2P energy trading simulations are performed with 282 participants on a modified 141-bus system \cite{mieth2019distribution}, and operation time is 1239 seconds for 51 rounds. This is shown that operation time for each round $n$ increases due to the number of participants. The results in Fig.~\ref{fig:social_welfare_141_rounds} reveal that when prosumers have more options to choose from, they tend to trade on lines with no congestion. As a result, the 141-bus system requires fewer rounds of iterations compared to the 15-bus system.

%new table
\begin{table}
% \vspace{0.5em}
\caption{Clearing prices $\lambda_{ij}$ between prosumer $i$ and its partners $j$.}
\label{tab:trade_prices}
\centering
\resizebox{1\columnwidth}{!}{%
\begin{tabular}{|c|c|c|c|c|c|c|c|c|c|c|c|c|c|c|}
\hline
\textbf{$\lambda_{ij}$} & \textbf{1} & \textbf{2} & \textbf{3} & \textbf{4} & \textbf{5} & \textbf{6} & \textbf{7} & \textbf{8} & \textbf{9} & \textbf{10} & \textbf{11} & \textbf{12} & \textbf{13} & \textbf{14} \\ \hline
\textbf{1}         & 9.62       & 9.62       & 9.62       & 12.66      & 0          & 12.66      & 9.62       & 9.62       & 9.62       & 9.62        & 9.62        & 9.62        & 9.62        & 9.62        \\ \hline
\textbf{2}         & 9.62       & 9.62       & 9.62       & 12.66      & 0          & 12.66      & 9.62       & 9.62       & 9.62       & 9.62        & 9.62        & 9.62        & 9.62        & 9.62        \\ \hline
\textbf{3}         & 0          & 0          & 0          & 0          & 0          & 0          & 0          & 0          & 0          & 0           & 0           & 0           & 0           & 0           \\ \hline
\textbf{4}         & 0          & 0          & 0          & 0          & 0          & 0          & 0          & 0          & 0          & 0           & 0           & 0           & 0           & 0           \\ \hline
\textbf{5}         & 0          & 0          & 0          & 15.71      & 0          & 15.71      & 0          & 0          & 0          & 0           & 0           & 0           & 0           & 0           \\ \hline
\textbf{6}         & 0          & 0          & 0          & 0          & 0          & 0          & 0          & 0          & 0          & 0           & 0           & 0           & 0           & 0           \\ \hline
\textbf{7}         & 0          & 0          & 0          & 0          & 0          & 0          & 0          & 0          & 0          & 0           & 0           & 0           & 0           & 0           \\ \hline
\textbf{8}         & 9.62       & 9.62       & 9.62       & 0      & 0          & 0      & 9.62       & 9.62       & 9.62       & 9.62        & 9.62        & 9.62        & 9.62        & 9.62        \\ \hline
\textbf{9}         & 0          & 0          & 0          & 0          & 0          & 0          & 0          & 0          & 0          & 0           & 0           & 0           & 0           & 0           \\ \hline
\textbf{10}        & 0          & 0          & 0          & 0          & 0          & 0          & 0          & 0          & 0          & 0           & 0           & 0           & 0           & 0           \\ \hline
\textbf{11}        & 0          & 0          & 0          & 0          & 0          & 0          & 0          & 0          & 0          & 0           & 0           & 0           & 0           & 0           \\ \hline
\textbf{12}        & 9.62       & 9.62       & 9.62       & 12.66      & 0          & 12.66      & 9.62       & 9.62       & 9.62       & 9.62        & 9.62        & 9.62        & 9.62        & 9.62        \\ \hline
\textbf{13}        & 0          & 0          & 0          & 0          & 0          & 0          & 0          & 0          & 0          & 0           & 0           & 0           & 0           & 0           \\ \hline
\textbf{14}        & 0          & 0          & 0          & 0          & 0          & 0          & 0          & 0          & 0          & 0           & 0           & 0           & 0           & 0           \\ \hline
\end{tabular}}
\end{table}

%new table
\begin{table}
\vspace{0.5em}
\caption{Clearing quantities $p_{ij}$ between prosumer $i$ and its partners $j$.}
\label{tab:trade_quantities}
\centering
\resizebox{1\columnwidth}{!}{%
\begin{tabular}{|c|c|c|c|c|c|c|c|c|c|c|c|c|c|c|}
\hline
\textbf{$p_{ij}$} & \textbf{1} & \textbf{2} & \textbf{3} & \textbf{4} & \textbf{5} & \textbf{6} & \textbf{7} & \textbf{8} & \textbf{9} & \textbf{10} & \textbf{11} & \textbf{12} & \textbf{13} & \textbf{14} \\ \hline
\textbf{1}         & 0.007      & 0.007      & 0.007      & 0.061      & 0          & 0.031      & 0.014      & 0.014      & 0.013      & 0.010       & 0.003       & 0.050       & 0.011       & 0.012       \\ \hline
\textbf{2}         & 0.005      & 0.005      & 0.005      & 0.012      & 0          & 0.113      & 0.002      & 0.002      & 0.002      & 0.004       & 0.004       & 0.241       & 0.003       & 0.003       \\ \hline
\textbf{3}         & 0          & 0          & 0          & 0          & 0          & 0          & 0          & 0          & 0          & 0           & 0           & 0           & 0           & 0           \\ \hline
\textbf{4}         & 0          & 0          & 0          & 0          & 0          & 0          & 0          & 0          & 0          & 0           & 0           & 0           & 0           & 0           \\ \hline
\textbf{5}         & 0          & 0          & 0          & 0.088      & 0          & 0.05      & 0          & 0          & 0          & 0           & 0           & 0           & 0           & 0           \\ \hline
\textbf{6}         & 0          & 0          & 0          & 0          & 0          & 0          & 0          & 0          & 0          & 0           & 0           & 0           & 0           & 0           \\ \hline
\textbf{7}         & 0          & 0          & 0          & 0          & 0          & 0          & 0          & 0          & 0          & 0           & 0           & 0           & 0           & 0           \\ \hline
\textbf{8}         & 0.003      & 0.003      & 0.003      & 0      & 0          & 0      & 0.006      & 0.006      & 0.006      & 0.004       & 0.002       & 0.001       & 0.005       & 0.005       \\ \hline
\textbf{9}         & 0          & 0          & 0          & 0          & 0          & 0          & 0          & 0          & 0          & 0           & 0           & 0           & 0           & 0           \\ \hline
\textbf{10}        & 0          & 0          & 0          & 0          & 0          & 0          & 0          & 0          & 0          & 0           & 0           & 0           & 0           & 0           \\ \hline
\textbf{11}        & 0          & 0          & 0          & 0          & 0          & 0          & 0          & 0          & 0          & 0           & 0           & 0           & 0           & 0           \\ \hline
\textbf{12}        & 0.005      & 0.005      & 0.005      & 0.012      & 0          & 0.024      & 0.002      & 0.002      & 0.002      & 0.004       & 0.004       & 0.330       & 0.003       & 0.003       \\ \hline
\textbf{13}        & 0          & 0          & 0          & 0          & 0          & 0          & 0          & 0          & 0          & 0           & 0           & 0           & 0           & 0           \\ \hline
\textbf{14}        & 0          & 0          & 0          & 0          & 0          & 0          & 0          & 0          & 0          & 0           & 0           & 0           & 0           & 0           \\ \hline
\end{tabular}}
\end{table}

%% new section
\subsection{P2P matching analysis and efficiency comparison}
Next we analyze social welfare of prosumers and the main grid with and without the two proposed dynamic pricing schemes. Since this is about P2P matching analysis, we consider the virtual layer only here, and the following P2P matching is divided into two cases.

\textbf{Case 1. Prosumers do not trade with main grid}:
When producers have sufficient energy to cover all consumers, there happens no grid trading; for example, this is the case during off-peak hours. As can be seen in Table~\ref{tab:Compare_with_central}, the proposed method effectively achieves the same social welfare as the centralized method without requiring all user information, as it shares user information only among matching prosumers.
\begin{table}[t]
% \vspace{0.5em}
\caption{Compared to the previous method in \textbf{Case 1}.}
\label{tab:Compare_with_central}
\centering
\resizebox{0.8\columnwidth}{!}{%
\begin{tabular}{|c|c|c|}
\hline
                   & \makecell{Social \\welfare (\textdollar)} & \makecell{Amount of exchanged\\ energy ($MWh$)}\\ \hline
Centralized method    & 22.182                & 1.2529                       \\ \hline
Proposed method & 22.182                & 1.2529                       \\ \hline
\end{tabular}}
\end{table}
%%% new table
\begin{table}[t]
\vspace{0.5em}
\caption{Compared to the previous method in \textbf{Case 2}.}
\label{tab:Compare_grid_trade}
\centering
\resizebox{\columnwidth}{!}{%
\begin{tabular}{|c|c|c|c|}
\hline
                        & UPS & DPS & {\cite{umer2021novel}} \\ \hline
\makecell{Generation cost of main grid (\textdollar)}       & 5.45        & 9.464       & 12.66     \\ \hline
\makecell{Exchanged energy with main grid ($MWh$)}      & 0.2162        & 0.3771        & 0.506    \\ \hline
\makecell{Welfare of consumers (\textdollar)}      & 11.234        & 11.296        & 11.346    \\ \hline
\makecell{Generation cost of producers (\textdollar)}      & 4.472        & 4.472        & 4.472    \\ \hline
\makecell{Exchanged energy among prosumers ($MWh$)}      & 0.3        & 0.3        & 0.3    \\ \hline
\makecell{Social welfare in centralized implementation (\textdollar)}      & 6.762        & 6.825      & 6.874    \\ \hline
\makecell{Social welfare in proposed method (\textdollar)}      & 6.762        & 6.825        & --    \\ \hline
\end{tabular}}
\end{table}

\textbf{Case 2. Prosumers trade with main grid}:
When producers do not have sufficient surplus energy to satisfy the consumer demand; for example, during on-peak hours, consumers have to purchase the remaining energy from the main grid. In doing this, we modify the generation parameters at nodes 5, 13, to 0.1$MWh$ and 0.2$MWh$, respectively, and the remaining nodes have zero $g_i^P$ and $g_i^Q$ in Table~\ref{tab:pro_pars}.

\subsubsection{Comparison between proposed community grid and traditional community}
As described in Section~\ref{subsec:PG_model}, the dynamic price is updated by prosumer demand. It offers the main grid flexibility to deal with the increase/decrease in demand. Table~\ref{tab:Compare_grid_trade} illustrates the different results in the proposed community grid and also compare with traditional community\cite{umer2021novel}. Note that the buying price from grid in \cite{umer2021novel} assumed is $25$\textdollar$/MWh$ and equal to minimum price $b_0$ of this study. Compared with \cite{umer2021novel}, the actual generated energy cost of main grid is decreased by 56.9\% and 25.2\% in UPS and DPS, respectively. Furthermore, our proposed method results in a reduction in grid trading in total, 57.3\% and 25.5\% in UPS and DPS, respectively. Although prosumers can adjust their consumption to maximize their benefits according to the main grid's price, they still experience a small negative impact on their welfare. The reason is that our method's welfare curve increases at a lower rate due to lower energy consumption, compared to the higher rate observed when energy consumption is higher \cite{umer2021novel}. Thus, \cite{umer2021novel} has social welfare slightly better than our proposed method, as can be seen in Table~\ref{tab:Compare_grid_trade}.
%%%table break
\begin{table}
% \vspace{0.5em}
\caption{P2P matching of proposed pricing schemes in \textbf{Case 2}.}
\label{tab:Compare_proposed_schemes}
\centering
\begin{adjustbox}{width=\columnwidth}{}
% \small
\begin{tabular}{|c|ccc|ccc|}
\hline
             & \multicolumn{3}{c|}{Unique price scheme (UPS)}                                                  & \multicolumn{3}{c|}{Differential price scheme (DPS)}                                            \\ \hline
\makecell{Prosumer\\No.} & \multicolumn{1}{c|}{$p_{0i}$} & \multicolumn{1}{c|}{$\sum_{j\in\omega_i} p_{ij}$} & \makecell{$W_i(d_i,g_i,p_{i0})$} & \multicolumn{1}{c|}{$p_{0i}$} & \multicolumn{1}{c|}{$\sum_{j\in\omega_i} p_{ij}$} & \makecell{$W_i(d_i,g_i,p_{i0})$} \\ \hline
1            & \multicolumn{1}{c|}{0.0201}    & \multicolumn{1}{c|}{0}             & 0.098                     & \multicolumn{1}{c|}{0.0201}   & \multicolumn{1}{c|}{0}             & 0.102                     \\ \hline
2            & \multicolumn{1}{c|}{0}         & \multicolumn{1}{c|}{0}             & 0                         & \multicolumn{1}{c|}{0}         & \multicolumn{1}{c|}{0}             & 0                         \\ \hline
3            & \multicolumn{1}{c|}{0}         & \multicolumn{1}{c|}{0}             & 0                         & \multicolumn{1}{c|}{0}         & \multicolumn{1}{c|}{0}             & 0                         \\ \hline
4            & \multicolumn{1}{c|}{0.0624}    & \multicolumn{1}{c|}{-0.0738}       & 1.906                     & \multicolumn{1}{c|}{0.1094}   & \multicolumn{1}{c|}{-0.0636}       & 1.666                     \\ \hline
5            & \multicolumn{1}{c|}{0}         & \multicolumn{1}{c|}{0.1000}           & -1.555                    & \multicolumn{1}{c|}{0}    & \multicolumn{1}{c|}{0.1000}           & -1.555                    \\ \hline
6            & \multicolumn{1}{c|}{0.0542}    & \multicolumn{1}{c|}{-0.1648}       & 7.418                     & \multicolumn{1}{c|}{0.0996}   & \multicolumn{1}{c|}{-0.1194}       & 6.284                     \\ \hline
7            & \multicolumn{1}{c|}{0}         & \multicolumn{1}{c|}{0}             & 0                         & \multicolumn{1}{c|}{0}         & \multicolumn{1}{c|}{0}             & 0                         \\ \hline
8            & \multicolumn{1}{c|}{0.0235}    & \multicolumn{1}{c|}{0}             & 0.232                     & \multicolumn{1}{c|}{0.0235}   & \multicolumn{1}{c|}{0}             & 0.2361                    \\ \hline
9            & \multicolumn{1}{c|}{0}         & \multicolumn{1}{c|}{0}             & 0                         & \multicolumn{1}{c|}{0}         & \multicolumn{1}{c|}{0}             & 0                         \\ \hline
10           & \multicolumn{1}{c|}{0}         & \multicolumn{1}{c|}{0}             & 0                         & \multicolumn{1}{c|}{0}         & \multicolumn{1}{c|}{0}             & 0                         \\ \hline
11           & \multicolumn{1}{c|}{0}         & \multicolumn{1}{c|}{0}             & 0                         & \multicolumn{1}{c|}{0.0132}   & \multicolumn{1}{c|}{0}             & 0.002                     \\ \hline
12           & \multicolumn{1}{c|}{0.0545}    & \multicolumn{1}{c|}{-0.0614}       & 1.581                     & \multicolumn{1}{c|}{0.1113}   & \multicolumn{1}{c|}{-0.1170}       & 3.006                     \\ \hline
13           & \multicolumn{1}{c|}{0}         & \multicolumn{1}{c|}{0.2000}           & -2.917                    & \multicolumn{1}{c|}{0}         & \multicolumn{1}{c|}{0.2000}           & -2.917                    \\ \hline
14           & \multicolumn{1}{c|}{0}         & \multicolumn{1}{c|}{0}             & 0                         & \multicolumn{1}{c|}{0}         & \multicolumn{1}{c|}{0}             & 0                         \\ \hline
\textbf{SUM} & \multicolumn{1}{c|}{0.2162}    & \multicolumn{1}{c|}{0}             & 6.762                     & \multicolumn{1}{c|}{0.3771}   & \multicolumn{1}{c|}{0}             & 6.825                     \\ \hline
\end{tabular}
\end{adjustbox}
\end{table}
\begin{figure}
% \vspace{0.5em}
\centering
\includegraphics[width=\columnwidth]{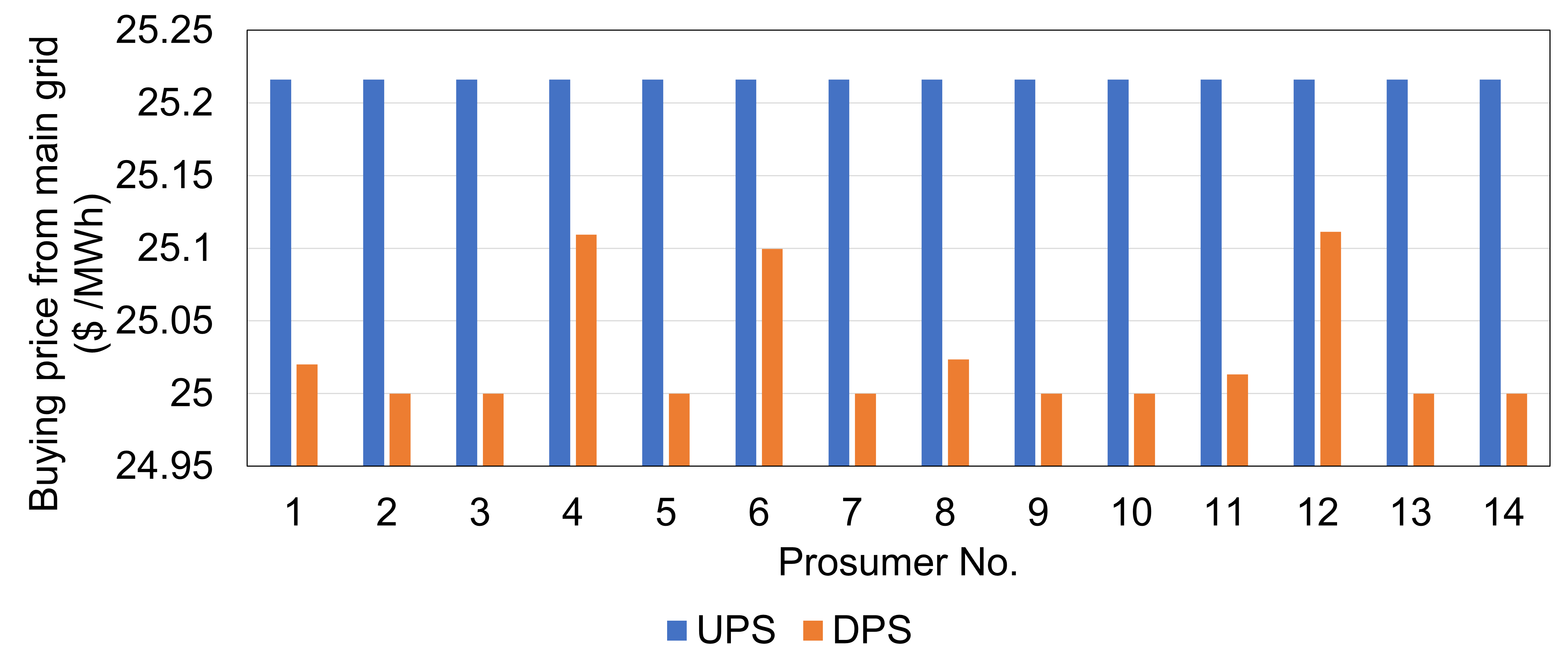}
\caption{Buying price from the main grid: $\overline\lambda$ (UPS), $\overline\lambda_i$ (DPS).}
\label{fig:compare_2_dynamics_price}
\end{figure}

In summary, the use of a quadratic cost function can reduce the generation cost of the main grid by 7.21\textdollar\ in UPS and 3.196\textdollar\ in DPS, compared to \cite{umer2021novel}. However, this lead to losses for consumers, with a welfare drop of 0.112\textdollar\ in UPS and 0.05\textdollar\ in DPS. This indicates a trade-off between market participants, prioritizing the grid's gains over consumer welfare. Consequently, the grid operator must thoroughly evaluate and balance the benefits and losses experienced by consumers while determining dynamic cost parameters $a_0$ and minimum price $b_0$ in the dynamic price model. This evaluation becomes crucial for incentivizing consumers to transition from the traditional market to the proposed market.

\subsubsection{Fairness of market participants in community grid}
In Section~\ref{subsec:PG_model}, we discussed how UPS and other schemes \cite{luo2021multiple,mehdinejad2022peer} are not fair to all prosumers. This subsection analyzes it with results. Each prosumer in the market is rationally deciding how much energy they wish to buy based on their satisfaction or preferences. Thus, the Jain's fairness index (JFI) in \cite{jain1984quantitative} is used to determine the fairness of UPS and DPS. (\ref{eq:fairness_index}) quantifies the spread of welfare among prosumers. JFI is higher with a fairer scheme than with a lower one.
\begin{equation}
JFI = \frac{{(\sum_{i\in\mathcal{N}} W_i)}^2}{|\mathcal{N}|\sum_{i\in\mathcal{N}} {(W_{i})}^2}. \label{eq:fairness_index}
\end{equation}
When calculating $JFI$ from Table~\ref{tab:Compare_proposed_schemes}, we discover that DPS is 7\% fairer than UPS. This is because prosumer 11 cannot purchase any energy in UPS, resulting in zero utilities and welfare. In DPS, this prosumer can purchase 0.0132$MWh$ of energy, yielding  0.332\textdollar\ in utilities and 0.002\textdollar\ in welfare. This comparison shows that prosumers in UPS are limited in their ability to purchase energy and experience losses due to the actions of other prosumers. In contrast, DPS does not have these limitations, which can encourage more prosumers to participate in the market. Fig.~\ref{fig:compare_2_dynamics_price} shows each prosumer's buying price from both the main grid for both UPS and DPS.
%% new section
\section{Conclusion}\label{sec:conclusion}
In this paper, we introduce a decentralized market framework for the P2P electricity market in both layers, leveraging the concept of distributed utility maximization to demonstrate optimality and the optimal solution. The virtual layer introduces a decentralized community grid, directly connecting prosumers to the main grid, thus enabling strategic price adjustments that minimize generation costs. By integrating congestion pricing, we harmonize virtual and physical requirements, resulting in improved efficiency for congestion management and P2P energy exchange. We assess our approach on a modified IEEE 15-bus system and 141-bus system and compare it to previous methods. Our results show the enhanced efficiency and effective congestion resolution, particularly in the 141-bus system. In future work, we plan to minimize uncertainty risks in hour-ahead to day-ahead electricity markets.

\renewcommand{\baselinestretch}{0.85}
\bibliographystyle{IEEEbib}

\end{document}